\documentclass[journal]{IEEEtran}
\ifCLASSINFOpdf
\else
\fi
\usepackage{cite}
\usepackage{setspace}
\usepackage{amsmath}
\usepackage{amsthm}
\usepackage{graphicx}
\usepackage{epstopdf}
\usepackage{amsfonts}
\newtheorem{remark}{Remark}

\usepackage{algorithm}
\usepackage{xcolor}
\usepackage[noend]{algpseudocode}

\usepackage{subcaption}
\usepackage{stfloats}
\usepackage{booktabs}

{ \bgroup
  \addtolength\abovedisplayshortskip{#1}
  \addtolength\abovedisplayskip{#1}
  \addtolength\belowdisplayshortskip{#1}
  \addtolength\belowdisplayskip{#1}}
{\egroup\ignorespacesafterend}

\newtheorem{proposition}{Proposition}

\newtheorem{theorem}{Theorem}
\newtheorem{itassumption}{Assumption}
\newtheorem{itassumptionr}{Assumption (rev.)}

\newtheorem{itresult}{Result}

\DeclareMathOperator{\diag}{diag}
\DeclareMathOperator{\rank}{rank}

\DeclareMathOperator{\sat}{sat}
\DeclareMathOperator{\sgn}{sgn}

\begin{document}
%\title{Adaptive control and disturbance decoupling}
\title{{\color{black}Resilience in Platoons of Cooperative Heterogeneous Vehicles: Self-organization Strategies and Provably-{\color{black}correct} Design}}


\author{%\IEEEauthorblockN{Di Liu}
Di Liu, Sebastian Mair, Kang Yang, Simone Baldi, Paolo Frasca, and Matthias Althoff  \vspace{-0.1cm}
%Di Liu
\thanks{This research was partly supported by the European Union Horizon 2020 research and innovation programme under
the Marie Sklodowska-Curie grant 899987, by the Natural
Science Foundation of China grant 62073074, and by the ANR (French National Science Foundation) no. ANR-18-
CE40-0010 (HANDY), and by the European Commission project justITSELF grant 817629.  (corresponding author: S. Baldi) \newline
D. Liu, S. Mair and M. Althoff are with School of Computation, Information and Technology, Technical University of Munich (TUM), Germany. D. Liu is also with Visual Intelligence for Transportation lab (VITA), Ecole Polytechnique Federale de Lausanne (EPFL), Switzerland (emails: \{di.liu, sebi.mair, althoff\}@tum.de) \newline 
%is with the School of Cyber Science and Engineering, Southeast
%University, Nanjing 210096, China, and also with the Bernoulli Institute for Mathematics, Computer Science and Artificial Intelligence, University of Groningen,  Groningen 9747AG, The Netherlands (email: di.liu@rug.nl) \newline  
K. Yang and S. Baldi are with Frontiers Science Center for Mobile Information Communication and Security, Southeast University, Nanjing, China. (emails: 220194629@seu.edu.cn, s.baldi@tudelft.nl)\newline 
P. Frasca is with University of Grenoble Alpes, CNRS, Inria, Grenoble INP, GIPSA-Lab, Grenoble F-38000, France (e-mail: paolo.frasca@gipsa-lab.fr)
}}
%\thanks{\\
%\\
%\\
%\\
%\\
%\\
%\\
%\\
%\\
%}
%}

%\markboth{Di Liu {Notes},~\emph{Mixed platoon}}%Journal of \LaTeX\ Class Files,~Vol.~14, No.~8, August~2018}%
%{Shell \MakeLowercase{\textit{et al.}}: Bare Demo of IEEEtran.cls for IEEE Transactions on Magnetics Journals}



% for Transactions on Magnetics papers, we must declare the abstract and
% index terms PRIOR to the title within the \IEEEtitleabstractindextext
% IEEEtran command as these need to go into the title area created by
\IEEEtitleabstractindextext{%
\begin{abstract}
This work proposes {\color{black}provably-correct self-organizing strategies for platoons of heterogeneous vehicles}. Self-organization is the capability to autonomously homogenize to a common group behavior. {\color{black}We show that self-organization keeps resilience to \emph{acceleration limits} and \emph{communication failures}, i.e., homogenizing to a common group behavior makes the platoon recover from these  impairments.} Adhering to {\color{black}acceleration} limits is achieved by self-organizing to a common constrained group behavior that prevents reaching acceleration limits. In the presence of communication failures, resilience is achieved by self-organizing to a common group observer to estimate the missing information. Stability {\color{black}and string stability} of the self-organization mechanism are studied analytically, {\color{black}and correctness with respect to {\color{black}traffic actions (e.g. emergency braking, cut-in, merging)} is realized through a provably-correct safety layer}. Numerical validations via the {\color{black}platooning toolbox OpenCDA in CARLA and via the CommonRoad platform} {\color{black} confirm improved performance through self-organization and the provably-correct safety layer}. 
\vspace{-0.3cm} \newline 
\end{abstract}
}

\maketitle

\IEEEdisplaynontitleabstractindextext

% \fi
\IEEEpeerreviewmaketitle

\section{Introduction}

Platooning refers to automated vehicles following each other at short distances by automatically accelerating and braking \cite{6338839}: pioneering studies date back to the 70's \cite{1100652}. Recognized advantages are better road exploitation and energy efficiency due to reduced air drag \cite{8671766,7437386}. After initially using only on-board radar sensing as in adaptive cruise control (ACC) technology, it was shown that the combination of on-board sensing and inter-vehicle wireless communication could achieve shorter inter-vehicle distances, leading to cooperative adaptive cruise control (CACC) \cite{6338839,6669336}. Inter-vehicle communication in CACC may involve the vehicle predecessor, the leading vehicle \cite{8049413}, or k-nearest neighboring vehicles \cite{7872392,9462542}. 

Resilience of CACC-equipped platoons, i.e., their capacity to recover from various impairments, has been the subject of several research efforts. \emph{Acceleration limits} might impair the formation-keeping task \cite{751767,1383798}, and approaches to improve resilience make use of cooperation mechanisms \cite{7972982} or adaptation mechanisms \cite{TAO20197,9007498}. \emph{Communication failures} also cause impairments in CACC: studies have shown that communication issues \cite{9037093,doi:10.1002/acs.3032,8360769} deeply affect CACC performance \cite{6907977}. Approaches to improve resilience include robust control \cite{9081994,8276639,9497782}, delay compensation \cite{8933478}, resource-efficient communication \cite{8606268,9099998}, dynamic topology adjustments \cite{7954618}, and observer design \cite{8550585}. While several approaches to improve resilience of CACC have been proposed, we are not aware of provably-correct CACC designs with respect to given specifications. {\color{black}Previous methods, e.g. fail-safe planners \cite{9091937}, positive invariant sets \cite{7995821,7944649} or control barrier functions \cite{8114339}, were only applied to ACC.} Specifications arising in lane changing \cite{7515222},  interaction with human drivers \cite{9072289}, or collision avoidance \cite{9001172,9210204}, are still not ensured in CACC. 

This work proposes \emph{self-organization} in CACC-equipped platoons of cooperative vehicles, to address causes of impairments like acceleration limits and communication failures. Our work covers the analytic study of stability for the self-organization mechanism and its provably-correct design. The term \emph{self-organization} refers to the capability of the platoon to homogenize to a-priori unknown common group dynamics, despite the vehicles in the platoon being heterogeneous. The reason for seeking a homogeneous behavior is that making all vehicles respond in an analogous way to impairments promotes resilience of the platoon. Seeds of this idea can be found in distributed model reference control \cite{8704911,8372462}, where heterogeneous systems are homogenized to the same reference dynamics. Notably, the common group behavior in our approach is not imposed a priori, but generated autonomously by vehicles composing the platoon. Our contributions are:
\begin{itemize}
\item[a)] We study analytically the stability {\color{black}and string stability} of the self-organization mechanism. While standard CACC \cite{7972982,9007498,6907977} does not achieve convergence of the spacing errors to zero when the platoon is composed of heterogeneous vehicles, self-organization realizes asymptotic convergence {\color{black}and string stable behavior}, despite the heterogeneity of the vehicles composing the platoon. 
\item[b)] We show that self-organization realizes resilience to acceleration limits. Here, the self-organization mechanism makes all vehicles converge autonomously to the worst-case acceleration limits in the platoon (constrained group model). 
Generating such group model prevents loss of cohesiveness and collisions, since the leading vehicle can limit its acceleration/deceleration to prevent all follower vehicles from hitting their acceleration limits. 
\item[c)] We show that self-organization realizes resilience to wireless communication failures. Here, the self-organization mechanism makes it possible to design an observer with common dynamics (group observer). With this observer, each vehicle reconstructs the missing signal to recover a performance close to the ideal communication scenario.
\item[d)] We combine CACC self-organization with a provably-correct safety layer for the first time. {\color{black}Our implementation %in the platooning toolbox OpenCDA in CARLA and a full-range implementation 
in CommonRoad integrates self-organization in a provably-correct safety layer to handle non-nominal operation.} %Our tests confirm the capability of the approach to recover the formation, % with convergence of the spacing errors, 
%even in the presence of the safety layer acting on top of the proposed self-organization mechanism. 
\end{itemize}

The rest of the paper is organized as follows: our solution concept is presented in Sec.~\ref{system_struc}. Self-organization is achieved and analyzed in Sec.~\ref{Chap_group}. Sec.~\ref{Chap_design} proposes resilient designs %to solve the impairments 
based on the self-organization idea. Effectiveness of the designs is validated in Sec.~\ref{Chap_simu}. Conclusions are drawn in Sec.~\ref{Chap_conclusion}.

\begin{figure}[t]
	\centering
	\includegraphics[%height=2.6cm,
	width=0.98\linewidth]{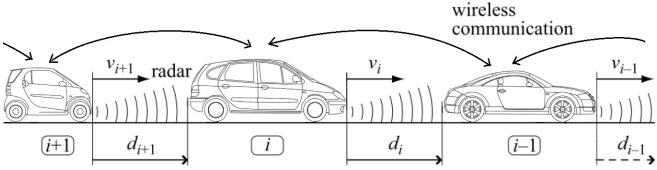}
	\caption{CACC-equipped platoon: on-board sensing of spacing error and relative velocity are supplemented by the wireless communication of additional variables (e.g. acceleration). }
	\label{platoon_fig}
\end{figure}

\section{CACC system structure and solution concept}\label{system_struc}
Consider a platoon composed of $M$ vehicles (Fig.~\ref{platoon_fig}). Let $S_M = \{ 2, 3, \ldots, M \}$ be the set of following vehicles, with $i=1$ reserved for the leading vehicle. {\color{black}In the CACC literature  
\cite{7972982,9007498,6907977}, a widely used model for the $i$-th vehicle is} %, $i \in \{1\} \cup S_M$
\begin{equation} \label{Math_platoon}
    \begin{bmatrix}
    \dot{q_i}(t)\vspace{0.15cm}\\
    \dot{v_i}(t)\vspace{0.15cm}\\
    \dot{a_i}(t)
    \end{bmatrix} = \begin{bmatrix}
    %v_{i-1}(t)-
		v_i(t)\vspace{0.05cm}\\
    a_i(t)\vspace{-0.25cm}\\
    -\frac{1}{\tau_i}a_i(t) +\frac{1}{\tau_i}(\overbrace{u_{bl_i}(t)+ u_{hm_i}(t)}^{u_{i}(t)})
    \end{bmatrix},
\end{equation}
where $q_i$ is the position of vehicle $i$, $v_i$ is its velocity, and $a_i$ is its acceleration. The control input $u_{i}$ to be designed is a desired acceleration, which differs from the acceleration $a_i$ due to the engine time constant $\tau_i>0$, representing the time required for the engine to reach the desired acceleration. Because each vehicle has its own engine time constant $\tau_i$, we propose the input $u_{i} = u_{bl_i}+ u_{hm_i}$, where  $u_{bl_i}$ is a \emph{baseline} input designed according to the standard CACC protocol (cf. Sec.~\ref{system_struc}-A), and $u_{hm_i}$ is an auxiliary input designed to \emph{homogenize} the vehicle behavior (cf. Sec.~\ref{system_struc}-B). The aim of homogenization is to promote resilience by letting all vehicles respond with the same dynamics to impairments.

\subsection{Baseline input}

The goal of each vehicle in the platoon, except the leading vehicle, is to maintain a desired distance with its precedessor. To this purpose, a constant time headway spacing policy is adopted \cite{7972982,9007498,6907977}, resulting in the spacing error 
\begin{align}\label{error}
e_i(t) &= q_{i-1}(t) - q_i(t) - h v_i (t), \nonumber \\%\label{error}
     %  &= (q_{i-1}(t) - q_i(t) - L_i ) - (r_i + h v_i (t))
\dot{e}_i(t) &= v_{i-1}(t) - v_i(t) - h a_i (t),\: \: i \in S_M, 
  %     &= (q_{i-1}(t) - q_i(t) - L_i ) - (r_i + h v_i (t))
\end{align}
where $h>0$ is the time headway constant and $h v_i$ is the velocity-dependent desired distance.  

The control objective is to regulate $e_i$ to zero. % for all $i \in S_M$. 
After defining the relative velocity $\nu_i = v_{i-1} - v_i$, the following dynamics are obtained from (\ref{Math_platoon}) and (\ref{error}): 
\begin{align} \label{LTI_platoon}
    \begin{bmatrix}
    \dot{e}_i(t)\\
    \dot{\nu}_i(t)\\
    \dot{a}_i(t)
    \end{bmatrix} &= \begin{bmatrix}
    0 & 1 & -h\\
    0 & 0 & -1\\
    0 & 0 & -\frac{1}{\tau_i}
    \end{bmatrix}
    \begin{bmatrix}
    e_i(t)\\ \nu_i(t)\\a_i(t)
    \end{bmatrix}  \\ 
		& \quad + \begin{bmatrix}
    0\\
    1\\
    0
    \end{bmatrix}
    a_{i-1}(t)
     + \begin{bmatrix}
    0\\0\\\frac{1}{\tau_i}
    \end{bmatrix}(u_{bl_i}(t) + u_{hm_i}(t)). \nonumber
\end{align}

{\color{black}The control input to regulate the spacing error in standard CACC literature \cite{7972982,9007498,6907977} uses  proportional-derivative action from $e_i$ and $\dot{e}_i$. A similar control is adopted here as a baseline input $u_{bl_i}$}: 
\begin{equation}
\label{cont}
\begin{split}
     h\dot{u}_{bl_i}(t) &= -u_{bl_i}(t) + \xi_{i}(t),  \\%, \;\; \forall i \in \{1\} \cup S_M \\ 
    \xi_{i}(t) &= \begin{cases} K_{p_i} e_i(t) + K_{d_i} \dot{e}_i(t) & \\
		  \quad + \xi_{hm_{i}}(t) + u_{bl_{i-1}}(t), \; & i \in S_M \\
    u_r(t) \; \;\;\;\;\;& i=1
    \end{cases}
\end{split}
\end{equation}
where $K_{p_i}$ and $K_{d_i}$ are the proportional and derivative gains, $u_{bl_{i-1}}$  is the cooperative signal received with wireless communication between vehicles $i$ and $i-1$, and $u_r$ is the desired acceleration of the leading vehicle. Compared to the standard control input in CACC, a new auxiliary term $\xi_{hm_i}$ is introduced in (\ref{cont}), whose aim is to homogenize the vehicle behavior. %, which is otherwise heterogeneous due to different gains $K_{p_i}$ and $K_{d_i}$. 
The homogenizing design for $u_{hm_i}$ and $\xi_{hm_i}$ is presented hereafter.

\subsection{Auxiliary homogenizing input}
To explain the homogenization idea, assume that, in place of the heterogeneous dynamics \eqref{LTI_platoon}-\eqref{cont}, each vehicle exhibits the following dynamics (referred to as \emph{homogeneous dynamics})
\begin{align} \label{LTI_platoon_hom}
    \begin{bmatrix}
    \dot{e}_i(t)\\
    \dot{\nu}_i(t)\\
    \dot{a}_i(t)
    \end{bmatrix} &= \begin{bmatrix}
    0 & 1 & -h\\
    0 & 0 & -1\\
    0 & 0 & -\frac{1}{\tau_0}
    \end{bmatrix}
    \begin{bmatrix}
    e_i(t)\\ \nu_i(t)\\a_i(t)
    \end{bmatrix} \nonumber \\
		& \qquad + \begin{bmatrix}
    0\\
    1\\
    0
    \end{bmatrix}
    a_{i-1}(t)
     + \begin{bmatrix}
    0\\0\\\frac{1}{\tau_0}
    \end{bmatrix}u_{bl_i}(t),
\end{align}
%and
\begin{equation}
\label{cont_hom}
%\begin{split}
     h\dot{u}_{bl_i}(t) = -u_{bl_i}(t) + K_{p_0} e_i(t) + K_{d_0} \dot{e}_i(t) + u_{bl_{i-1}}(t).
\end{equation}
The difference compared to \eqref{LTI_platoon}-\eqref{cont} is that $\tau_0$, $K_{p_0}$ and $K_{d_0}$ are homogeneous for all vehicles. Recall that homogeneous vehicle dynamics make it easier to achieve synchronized platooning behavior \cite{7972982,9007498,6907977}. A natural question arises: \emph{how to converge to the homogeneous dynamics} \eqref{LTI_platoon_hom}-\eqref{cont_hom} starting from the heterogeneous dynamics \eqref{LTI_platoon}-\eqref{cont}? We answer this question by establishing the following result:

\begin{proposition}
%For any vehicle $i$, t
The auxiliary inputs $u_{hm_i}$ in \eqref{LTI_platoon} and $\xi_{hm_i}$ in \eqref{cont} {\color{black}homogenizing} any vehicle $i$ to dynamics \eqref{LTI_platoon_hom} and \eqref{cont_hom} are
\begin{align}
\label{gain_hom}
%\begin{split}
u_{hm_i}(t) &= \frac{\tau_0-\tau_i}{\tau_0}(a_i(t)-u_{bl_i}(t)), \nonumber \\
\xi_{hm_i}(t) &= (K_{p_0}-K_{p_i}) e_i(t) +(K_{d_0}-K_{d_i}) \dot{e}_i(t). %(-v_i(t) - h a_i(t) + v_{i+1}(t))
\end{align}
\end{proposition}
\begin{proof}
This can be verified by direct substitution of \eqref{gain_hom} into \eqref{LTI_platoon}-\eqref{cont}, which results in the dynamics \eqref{LTI_platoon_hom}-\eqref{cont_hom}. 
\end{proof}

\begin{remark}[\emph{A-priori knowledge}]
The homogenizing inputs \eqref{gain_hom} require a common set of parameters $(\tau_0, K_{p_0}, K_{d_0})$. The main idea of this work, presented in Sec.~III, is to avoid imposing common parameters a priori, but generate a \emph{homogeneous group behavior} autonomously.\vspace{-0.2cm}
\end{remark}

\begin{figure*}[t]
      \centering
\begin{subfigure}{0.48\textwidth}
      \centering
      \vspace{-0.3cm}\includegraphics[width=1.05\textwidth]{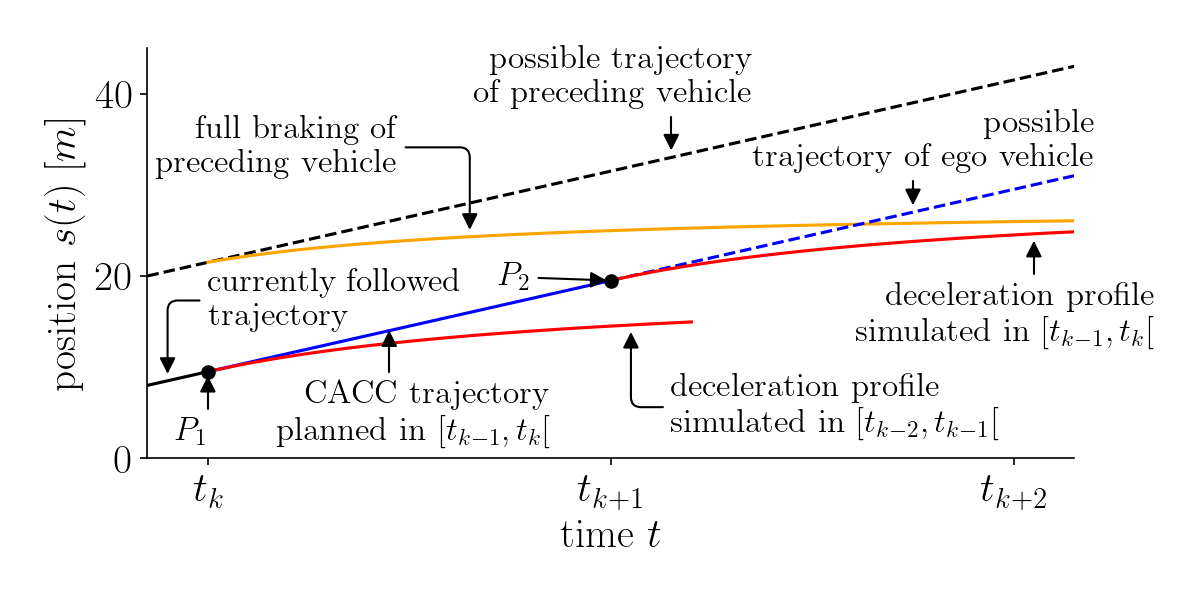}
  \caption{Interaction between the nominal and the safety controllers.}
  \label{cacc-flow-chart-fail}
\end{subfigure}
    ~ %add desired spacing between images, e. g. ~, \quad, \qquad, \hfill etc. 
\begin{subfigure}{0.48\textwidth}
  \centering
  % include second image
      \vspace{-0.3cm}\includegraphics[width=0.835\textwidth]{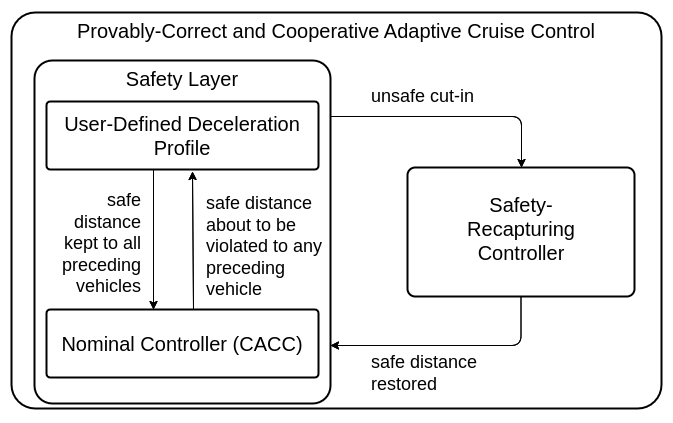}
  \caption{{\color{black}Statechart} of the provably-correct control layer.}
  \label{cacc-flow-chart}
\end{subfigure}
\caption{Concept of the provably-correct architecture.}
\label{TOT-cacc-flow-chart}
\end{figure*}

\subsection{Provably-correct safety layer} \label{Chap_safety}

To ensure safety of the CACC, we embed it into a provably-correct safety layer in line with \cite{9091937}. The goal is to keep a certain distance to the preceding vehicle such that safe emergency braking is always possible, even if a preceding vehicle fully brakes. The concept of the provably-correct layer is illustrated in Fig.~\ref{cacc-flow-chart-fail}, which we explain subsequently. We assume that the CACC protocol computes inputs $u_i$ for discrete, equidistant points in time $t_k=k dt$, where $k \in \mathbb{N}$ is the time step and $dt$ is the sampling time. In the time interval $[t_{k-1}, t_k[$, the CACC computes an input to be applied from time $t_k$ onward. This input is simulated for the ego vehicle until time $t_{k+1}$, followed by a user-defined emergency deceleration profile until standstill. At the same time, a forward simulation of the preceding vehicle is executed, assuming it applies the largest possible deceleration. If no collision between the ego vehicle and the preceding vehicle occurs, the input computed by the CACC is applied during the interval $[t_k,t_{k+1}[$. Otherwise, the emergency deceleration profile starts to be applied from $t_k$ onward, the safety of which was verified during the interval $[t_{k-2}, t_{k-1}[$.

The architecture also includes a safety-recapturing controller, designed for the case of cut-ins by other vehicles violating the safe distance between the ego vehicle and the cut-in vehicle: the recapturing controller ensures that the safe distance is re-established within a specified time period, assuming that the cut-in vehicle does not exceed a certain deceleration. We would like to refer the reader to the concept in Fig.~\ref{cacc-flow-chart} or to \cite{9091937} for details, {\color{black}including details on how dynamical properties of monotone systems %that hold for the vehicle dynamics under consideration, allow to dramatically 
make it possible to dramatically alleviate computational complexity}. The safety layer is implemented in the CommonRoad platform \cite{7995802}, as clarified in Sec.~V.

\subsection{Problem formulation}
The goal of this work is to achieve resilience via self-organization. Let us formulate the objectives as follows:

\underline{\emph{Self-organization Problem}} (group behavior): Consider the dynamics \eqref{LTI_platoon}, \eqref{cont} in the presence of heterogeneous engine constants $\tau_i$ and heterogeneous proportional-derivative gains $K_{p_i}$, $K_{d_i}$. Design the auxiliary inputs $u_{hm_i}$ and $\xi_{hm_i}$ in \eqref{LTI_platoon}, \eqref{cont} where, differently from Proposition 1, homogenization to common dynamics $\tau_0$, $K_{p_0}$, $K_{d_0}$ is achieved {\color{black}autonomously} without such dynamics being known a priori. 

\underline{\emph{Resilience Problem 1}} (acceleration limits): Consider the dynamics \eqref{LTI_platoon} in the presence of possibly heterogeneous acceleration limits $a_i(t) \in [a_{min_i},a_{max_i}]$. Design a suitable group behavior where each vehicle in the platoon can still maintain the formation without hitting the acceleration limits.

\underline{\emph{Resilience Problem 2}} (communication failures): Consider the control action \eqref{cont} with missing $u_{bl_{i-1}}$ due to communication failures. Design a suitable group behavior where each vehicle in the platoon can still maintain the formation by reconstructing the missing $u_{bl_{i-1}}$.

\section{Converging to a group behavior} \label{Chap_group}

This section contains the main result about convergence of the platoon dynamics to a group behavior, and stability properties of the resulting system. Before this, a few preliminaries on consensus dynamics are introduced.  

\subsection{Preliminaries on consensus dynamics}
To explain how to converge to a common group model, let us define the  consensus dynamics
\begin{align}
%\label{cons_hom}
\dot{\bar{K}}_{p \tau_i}(t) &= \mu_p \sum_{j \in \mathcal{N}_i} %a_{ij} 
(\bar{K}_{p \tau_j}(t) - \bar{K}_{p \tau_i}(t)), %\quad \bar{K}_{p \tau_i}(0) = K_{p \tau_i} 
\nonumber 
\end{align}
\begin{align}
\label{cons_hom}
\dot{\bar{K}}_{d_i}(t) &= \mu_d \sum_{j \in \mathcal{N}_i} %a_{ij} 
(\bar{K}_{d_j}(t) - \bar{K}_{d_i}(t)), %\quad \bar{K}_{d_i}(0) = K_{d_i}
  \nonumber \\
\dot{\bar{\tau}}_{i}(t) &= \mu_{\tau} \sum_{j \in \mathcal{N}_i} %a_{ij} 
(\bar{\tau}_{j}(t) - \bar{\tau}_{i}(t)),  %\quad  \bar{\tau}_{i}(0) = \tau_{i}
\end{align}
with initial conditions 
\begin{align}\label{prod_K}
\bar{K}_{p \tau_i}(0) = K_{p \tau_i}:=K_{p_i} \tau_i, \  \bar{K}_{d_i}(0) = K_{d_i}, \ \bar{\tau}_{i}(0) = \tau_{i}, 
\end{align}
where $\mu_p$, $\mu_d$, $\mu_{\tau}>0$ are consensus gains, and $\mathcal{N}_i$ indicates {\color{black}%the neighbors of vehicle $i$, i.e. 
the set of vehicles that vehicle $i$ can communicate with}. The variables $\bar{K}_{p \tau_i}$, $\bar{K}_{d_i}$, $\bar{\tau}_{i}$ are \emph{consensus variables}, whose aim is to converge to common values for all vehicles. It is well known that convergence depends on the graph connectivity \cite{frasca_book}, in line with the following result.  

\begin{proposition}[Average consensus \cite{frasca_book}]
For an undirected connected communication graph, the consensus dynamics \eqref{cons_hom} leads to
\begin{align}
\label{cons_hom_b}
\lim_{t \rightarrow \infty} \bar{K}_{p \tau_i}(t) &= \frac{1}{M}\sum_{j=1}^M K_{p \tau_j}=:K_{{p \tau}_0},  \nonumber \\
\lim_{t \rightarrow \infty} \bar{K}_{d_i}(t) &= \frac{1}{M}\sum_{j=1}^M K_{d_j}=:K_{d_0}, \nonumber \\
\lim_{t \rightarrow \infty} \bar{\tau}_{i}(t) &= \frac{1}{M}\sum_{j=1}^M \tau_j=:\tau_{0},
\end{align}
exponentially.   
\end{proposition}
The communication graph in Proposition 2 should not be confused with the inter-vehicle look-ahead interaction graph due to on-board sensing. For instance, an undirected communication graph arises from  bidirectional wireless communication often considered in CACC \cite{9007498,8933478,7525305} where $\mathcal{N}_i$ comprises the predecessor and the follower of each vehicle.

\subsection{Self-organizing design and stability analysis}

The parameters in \eqref{cons_hom_b} have been denoted as $\tau_0$, $K_{p_0}$, $K_{d_0}$ like the common set of parameters in \eqref{LTI_platoon_hom}-\eqref{cont_hom}. Note, however, that the parameters in \eqref{cons_hom_b} are a-priori unknown to all vehicles in the platoon. In line with \cite{8704911}, such parameters define an autonomously-generated common group behavior. To let all vehicles in the platoon homogenize to such common behavior, the following design for the auxiliary inputs is proposed:
\begin{align}
\label{gain_hom2}
%\begin{split}
u_{hm_i}(t) &= \frac{\bar{\tau}_{i}(t)-\tau_i}{\bar{\tau}_{i}(t)}(a_i(t)-u_{bl_i}(t)), \nonumber \\
\xi_{hm_i}(t) &= (\bar{K}_{p_i}(t)-K_{p_i}) e_i(t) +(\bar{K}_{d_i}(t)-K_{d_i}) \dot{e}_i(t) %(-v_i(t) - h a_i(t) + v_{i+1}(t))
\end{align}
where $\bar{K}_{p_i}(t) = \bar{K}_{p \tau_i}(t)/\bar{\tau}_{i}(t)$. The main difference between \eqref{gain_hom} and \eqref{gain_hom2} is replacing the constants $K_{p_0}$, $K_{d_0}$, $\tau_0$ with the consensus variables $\bar{K}_{p_i}(t)$, $\bar{K}_{d_i}(t)$, $\bar{\tau}_{i}(t)$. A non-trivial question is to analyze the stability properties resulting from \eqref{gain_hom2}: in fact, the gains in \eqref{gain_hom2} are time-varying due to \eqref{cons_hom} and result in a nonlinear closed-loop system. 

The following  result shows the stability properties of the self-organizing strategy \eqref{cons_hom}, \eqref{gain_hom2}, i.e. the capability of achieving asymptotic formation thanks to the common group behavior. {\color{black}Later, string stability is elaborated.}

\begin{theorem}
Consider the platoon composed of vehicle error dynamics \eqref{LTI_platoon}  and input dynamics  \eqref{cont}. Select the auxiliary inputs as in \eqref{gain_hom2} with parameters $\bar{\tau}_{i}$, $\bar{K}_{p_i}$, $\bar{K}_{d_i}$ arising from the consensus dynamics \eqref{cons_hom}. Then, the resulting spacing error $e_i$, $i \in S_M$ converge to zero as time goes to infinity. 
\end{theorem}

\begin{proof}
The proof relies on showing that, although the closed-loop system formed by \eqref{LTI_platoon}, \eqref{cont}, \eqref{cons_hom}, and \eqref{gain_hom2} is nonlinear, it can be written in the form $\dot{x}(t) = F x(t) + \tilde{F}(t) x(t) + \hat{F}(t)$, where $F$ is a Hurwitz state matrix, and $\tilde{F}(t)$, $\hat{F}(t)$ are exponentially vanishing nonlinear terms acting as disturbances (the state $x$ will be defined later). For this type of %perturbed linear 
dynamics, one can use the result in \cite[Lemma 1]{8704911} stating that the exponentially vanishing terms still allow exponential stability of state $x$.

First, consider vehicle 1: because vehicle 1 has no preceding vehicle, {\color{black}its dynamics do not contain any spacing error $e_1$ and relative velocity $\nu_1$. Only the dynamics of $x_1=[a_1\ u_{bl_1}]^T$ must be calculated. By using the last equations in  \eqref{LTI_platoon}, \eqref{cont}, and the first equation in \eqref{gain_hom2}, we obtain the dynamics}
\begin{align} \label{LTI_platoon_hom37}
\dot{x}_1(t) = F_1 x_1(t) + \tilde{F}_1(t) x_1(t) + \hat{F}_1(t), 
\end{align}
with the Hurwitz state matrix
\begin{align} \label{LTI_platoon_hom37b}
 F_1 = \begin{bmatrix}
    -\frac{1}{\tau_0}& \frac{1}{\tau_0}\\
     0 & -\frac{1}{h}
    \end{bmatrix}
\end{align}\vspace{-0.15cm}
and\vspace{-0.05cm} 
\begin{align} \label{LTI_platoon_hom38}
\tilde{F}_1(t) = \begin{bmatrix}
    -\frac{\tilde{\tau}_1(t)}{\tau_0(\tau_0-\tilde{\tau}_1(t))} &  \frac{\tilde{\tau}_1(t)}{\tau_0(\tau_0-\tilde{\tau}_1(t))}\\
    0 & 0
    \end{bmatrix}, \ \hat{F}_1(t) = \begin{bmatrix}
    0\\
    \frac{u_{r}(t)}{h}
    \end{bmatrix}.
\end{align}\vspace{-0.15cm}

For the other vehicles, we define $x_i=[e_i\  \nu_i\ a_i\ u_{bl_i}]^T$ and proceed by combining the platooning system \eqref{LTI_platoon}-\eqref{cont} with the auxiliary inputs \eqref{gain_hom2}, so as to obtain %the form %we obtain the dynamics
\begin{align} \label{LTI_platoon_hom37}
\dot{x}_i(t) = F x_i(t) + \tilde{F}_i(t) x_i(t) + \hat{F}_i(t), 
\end{align}
with\vspace{-0.3cm} %the state matrix
\begin{align} \label{LTI_platoon_hom37a}
F \hspace{-0.075cm}=\hspace{-0.075cm} \begin{bmatrix}
    0 & 1 & -h & 0\\
    0 & 0 & -1& 0\\
    0 & 0 & -\frac{1}{\tau_0}& \frac{1}{\tau_0}\\
    \frac{K_{p_0}}{h} & \frac{K_{d_0}}{h} & -K_{d_0}& -\frac{1}{h}
    \end{bmatrix}, 
\end{align}\vspace{-0.3cm}
%and\vspace{-0.3cm} %the terms  
\begin{align} \label{LTI_platoon_hom37b}
\tilde{F}_i(t) \hspace{-0.075cm}=\hspace{-0.075cm} \begin{bmatrix}
    0 & 0 & 0 & 0\\
    0 & 0 & 0 & 0\\
    0 & 0 & \hspace{-0.1cm}\frac{-\tilde{\tau}_i(t)}{\tau_0(\tau_0-\tilde{\tau}_i(t))}\hspace{-0.075cm} & \hspace{-0.075cm}\frac{\tilde{\tau}_i(t)}{\tau_0(\tau_0-\tilde{\tau}_i(t))}\hspace{-0.025cm}\\
    \frac{\tilde{K}_{p_i}(t)}{h} &  \frac{\tilde{K}_{d_i}(t)}{h} & -\tilde{K}_{d_i}(t)  &0 
    \end{bmatrix}\hspace{-0.075cm}, \ 
\end{align}
and $\hat{F}_i(t) = [ 0 \ a_{i-1}(t) \  0\ \frac{u_{bl_{i-1}}(t)}{h}]^T$. Note that $F$ is homogeneous for all vehicles and is Hurwitz as a result of the Routh-Hurwitz conditions \cite[Sec.~5.5.3]{goodwin2001}
\begin{equation}\label{condtt}
 h > 0,\quad K_{p_0},K_{d_0} > 0,\quad K_{d_0}> \tau_0 K_{p_0}.
\end{equation}
We refer the reader to Remark 2 for details on why \eqref{condtt} holds for the group parameters, given that the controller of each vehicle is stable, i.e. given that the Routh-Hurwitz conditions $K_{p_i},K_{d_i} > 0$, $K_{d_i}> \tau_i K_{p_i}$ hold individually.

To derive exponential stability from \cite[Lemma 1]{8704911}, let us start with \eqref{LTI_platoon_hom37} in the absence of exogenous input $u_r$, i.e. $\hat{F}_1=0$. Noticing that the term $\tilde{\tau}_1 = \tau_0-\tau_1$ converges to zero exponentially due to Proposition 2, we have that $\tilde{F}_i$ in \eqref{LTI_platoon_hom38} is exponentially vanishing. Thus, exponential stability of $x_1$ is concluded. At this point, one can proceed iteratively from \eqref{LTI_platoon_hom37} for all vehicles $i=2,3,\ldots$. We have that $\tilde{F}_i$ in \eqref{LTI_platoon_hom37b} converges to zero exponentially because $\tilde{\tau}_i$, $\tilde{K}_{p_i}$, $\tilde{K}_{d_i}$ converge to zero exponentially due to Proposition 2. Meanwhile, $\hat{F}_i$ also converges to zero exponentially because exponential convergence of $a_{i-1}$, $u_{bl_{i-1}}$ has been previously shown. 
	
In other words, we obtain that the preceding vehicle $i-1$ affects vehicle $i$ with an exponentially vanishing disturbance, leading to exponential stability of $x_i$ for any vehicle in the platoon, which concludes the proof.
\end{proof}

\begin{remark}[{\color{black}\emph{Stability and string stability}}]
{\color{black}Note that, when each vehicle satisfies the Routh-Hurwitz condition  $K_{d_i} > \tau_{i} K_{p_i} = K_{p \tau_i}$, we have $\sum_i K_{d_i} > \sum_i K_{p \tau_i}$, i.e. %with $K_{p \tau_i}$ in \eqref{prod_K} 
the Routh-Hurwitz condition is satisfied on average. %$\avg[K_{d_i}] > \avg[K_{p \tau_i}]$, and 
A common stabilizing gain for the group model is obtained via $\bar{K}_{p_i}(t) = \bar{K}_{p \tau_i}(t)/\bar{\tau}_{i}(t)$. In addition, as the transfer function from $a_{i-1}$ to $a_{i}$ for homogeneous vehicle dynamics is $\Gamma(s) = 1/(hs+1)$ \cite{7954618},  we achieve the standard notion of string stability
\begin{equation}\label{string_stab}
%\Gamma(s) = \frac{1}{hs+1}  \ \Rightarrow \  
\left| \Gamma(j \omega) \right| \leq 1, \forall \omega \geq 0.
\end{equation}
In other words, convergence to a homogeneous group model leads to a stable and string stable behavior.}
\end{remark}

\section{Self-organization for resilient design}\label{Chap_design}

Having shown the properties of self-organization, we are in a position to use group behavior ideas to solve the two resilience problems in Sec.~II-D. Specifically, we will design a \emph{constrained group model} to face acceleration limits, and a \emph{group observer model} to face communication failures.

\subsection{Resilient design to {\color{black}acceleration} limits}
We refer the reader to Sec.~V-B for an example of the problems that acceleration limits can create. Denote with $a_{max_i}$  the maximum acceleration of vehicle $i$ and with $a_{min_i}$ its minimum acceleration (maximum deceleration). We will design a resilient mechanism based on a simple idea: limiting the behavior of the leading vehicle (i.e. limit its acceleration/deceleration) avoids that the follower vehicles hit their limits. Stemming from this idea, we need to find the worst-case acceleration limits in the platoon, i.e. the minimum among the maximum accelerations, and the maximum among the minimum accelerations. To this purpose, we use the max-min consensus \cite{CORTES2008726}, aiming to find the minimum and the maximum value in a network in a distributed way 
\begin{align}\label{cons_hom2}
\dot{\bar{a}}_{max_i}(t) = \sgn_-\left( \sum_{j \in \mathcal{N}_j} \left( \bar{a}_{max_j}(t)-\bar{a}_{max_j}(t)\right)\right) \nonumber \\ %& \quad \bar{a}_{min_i}(0)=a_{min_i} 
\dot{\bar{a}}_{min_i}(t) = \sgn_+\left( \sum_{j \in \mathcal{N}_j} \left( \bar{a}_{min_j}(t)-\bar{a}_{min_i}(t)\right)\right), 
\end{align}
where
\begin{align}\label{cons_hom22}
\sgn_+(x) = \begin{cases} 0 &\hspace{-0.2cm}{\rm if} \;  x \leq 0 \\
			1   &\hspace{-0.2cm}{\rm if} \;  x > 0 \end{cases}, \ 
			\sgn_-(x) = \begin{cases} 0 &\hspace{-0.2cm}{\rm if} \;  x \geq 0 \\
			1   &\hspace{-0.2cm}{\rm if} \;  x < 0 \end{cases}
\end{align}
and with initial conditions $$\bar{a}_{max_i}(0)=a_{max_i},\quad  \bar{a}_{min_i}(0)=a_{min_i}.$$ 
In other words, $\bar{a}_{max_i}$ and $\bar{a}_{min_i}$ represent the estimates of the maximum and minimum acceleration in the platoon according to vehicle $i$. Clearly, we want such estimates to converge consistently along the platoon, i.e. $\bar{a}_{max_i}$ should converge to a common value for all vehicles, and analogously for $\bar{a}_{min_i}$.  The following convergence result holds. 

\begin{proposition}[Max-min consensus \cite{CORTES2008726}]
For an undirected connected communication graph, the dynamics \eqref{cons_hom2} lead to
\begin{align}
\label{maxmin_cons_hom_b}
\lim_{t \rightarrow \infty} \bar{a}_{max_i}(t) &= \min_i\left[ a_{max_i} \right]=: a_{max_0} \nonumber \\
\lim_{t \rightarrow \infty} \bar{a}_{min_i}(t) &= \max_i\left[ a_{min_i} \right]=: a_{min_0}
\end{align}
in finite time. 
\end{proposition}

\begin{remark}[No a-priori knowledge of the group model]
Albeit having different dynamics, the consensus mechanisms \eqref{cons_hom} and \eqref{cons_hom2} obey a similar philosophy: the variables %$(\bar{\tau}_{i},\bar{K}_{p\tau_i},\bar{K}_{d_i})$ or 
$(\bar{a}_{min_i},\bar{a}_{max_i})$ are initialized according to the vehicle parameters % ($(\tau_{i},K_{p\tau_i},K_{d_i})$ or $(a_{min_i},a_{max_i})$), 
and shared in a distributed way among neighboring vehicles, until convergence to common values $a_{max_0}$, $a_{min_0}$ in \eqref{maxmin_cons_hom_b}. These common values, which are a-priori unknown, % {\color{black}which are a-priori unknown. Such a-priori unknown values} 
can be regarded as an autonomously {\color{black}generated} \emph{constrained common group model}.
\end{remark}

Once the worst-case acceleration limits in the platoon are available, the desired acceleration can be limited within the same acceleration constraints arising from the constrained common group model. Specifically, %show that the acceleration is the result of a first-order filter from the desired acceleration, 
one can modify the control input \eqref{cont} as: 
\begin{equation}\label{awu}
h \dot{u}_{bl_i} = \begin{cases} 0 &{\rm if} \  u_{i} \hspace{-0.05cm}=\hspace{-0.05cm} \bar{a}_{max_i} \ {\rm and} \ \hspace{-0.05cm}-\hspace{-0.05cm}u_{bl_i} \hspace{-0.05cm}+\hspace{-0.05cm} \xi_{i}\geq 0\vspace{-0.3cm}\\ \\
			-u_{bl_i} + \xi_{i}   & {{\rm if} \  \bar{a}_{min_i} \hspace{-0.05cm}<\hspace{-0.05cm} u_{i} \hspace{-0.05cm}<\hspace{-0.05cm} \bar{a}_{max_i} }\\
		 & \ {\rm or} \  u_{i} \hspace{-0.05cm}=\hspace{-0.05cm} \bar{a}_{max_i} \ {\rm and} \ \hspace{-0.05cm}-\hspace{-0.05cm}u_{bl_i} \hspace{-0.05cm}+\hspace{-0.05cm} \xi_{i}< 0\\ 
		  & \ {\rm or} \  u_{i} \hspace{-0.05cm}=\hspace{-0.05cm} \bar{a}_{min_i} \ {\rm and} 	 \ -u_{bl_i} \hspace{-0.05cm}+\hspace{-0.05cm} \xi_{i} > 0\vspace{-0.3cm}\\ \\
			0 & {\rm if} \  u_{i} \hspace{-0.05cm}=\hspace{-0.05cm} \bar{a}_{min_i} \ {\rm and} \ \hspace{-0.05cm}-\hspace{-0.05cm}u_{bl_i} \hspace{-0.05cm}+\hspace{-0.05cm} \xi_{i}\leq 0. \end{cases}
\end{equation}

\noindent Analogously, since the worst-case bounds $\bar{a}_{max_i}$, $\bar{a}_{min_i}$ can be applied also to the acceleration, one can constrain the acceleration within the same limits, that is $a_i \in [\bar{a}_{min_i}, \bar{a}_{max_i}]$. The proposed resilient approach can be sketched {\color{black} via the following procedure}:

\begin{algorithm}
\floatname{algorithm}{Resilient Strategy}
\caption{(against acceleration limits)}
\label{array-sum0}
\begin{algorithmic}[1]
\Require initial conditions $a_{max_i}$, $a_{min_i}$. 
\Ensure the constrained control \eqref{awu} with constrained acceleration $a_i \in [\bar{a}_{min_i}, \bar{a}_{max_i}]$.  \vspace{0.15cm}
\State Online, at time $t$: Run the max-min consensus dynamics  \eqref{cons_hom2} to receive $\bar{a}_{min_i}(t)$, $\bar{a}_{max_i}(t)$; % converge to the a-priori unknown worst-case limits $a_{min_0}$, $a_{max_0}$ in \eqref{maxmin_cons_hom_b};
\State Online, at time $t$: use $\bar{a}_{min_i}(t)$, $\bar{a}_{max_i}(t)$  to realize the constrained input \eqref{awu} and the constrained acceleration $a_i \in [\bar{a}_{min_i}(t), \bar{a}_{max_i}(t)]$.
\end{algorithmic}
\end{algorithm}

\subsection{Resilient design to communication failures}

%We refer the reader to Sec.~V-B for an example showing the adverse effect of losing communication. 
Without communication, the CACC law in \eqref{cont} cannot make use of $u_{bl_{i-1}}$, i.e., the control input degenerates to a standard ACC \cite{6338839,6669336} where only on-board radar sensing is used:
\begin{equation}\label{cont_hom_nocomm}
h\dot{u}_{bl_i}(t) = -u_{bl_i}(t) + K_{p_0} e_i(t) + K_{d_0} \dot{e}_i(t). 
\end{equation}
We refer the reader to Sec.~V-B for an example showing the adverse effects of losing communication. To obtain resilience, we propose an observer to estimate the missing $u_{bl_{i-1}}$. As communication is inactive, let us assume that all vehicles already homogenized by means of \eqref{cons_hom} when communication was active, so that the vehicle dynamics have converged to $K_{p_0}$, $K_{d_0}$, $\tau_0$ in \eqref{cons_hom_b}. The observer design follows three steps: 
\begin{itemize}
\item[1)] Forming a model for the homogenized dynamics of vehicle $i-1$, i.e. the preceding vehicle; 
\item[2)] Designing a high-gain observer for some auxiliary variables of the homogenized preceding vehicle; 
\item[3)] Designing a sliding-mode observer for the missing %signal 
$u_{bl_{i-1}}$.
\end{itemize}
Thanks to self-organization \eqref{cons_hom}, all preceding vehicles $i-1$ exhibit the same homogenized dynamics: as a result, all vehicles design the same observer based on the same preceding vehicle dynamics. This concept would not be applicable in ACC with heterogeneous vehicles, unless assuming the platoon to be homogeneous or proceeding at constant speed \cite{8550585,9497782}. The observer design below does not require these assumptions.   

% \cite{6338839,6669336}, and it is the main reason why observer design is not considered in standard ACC, unless all vehicles are assumed

\subsubsection{Dynamics of the homogenized preceding vehicle} As all vehicles have converged to homogeneous dynamics via \eqref{cons_hom}, using 
%vehicle $i-1$ exhibits homogeneous dynamics described by $K_{p_0}$, $K_{d_0}$, $\tau_0$. Using 
\eqref{LTI_platoon_hom}-\eqref{cont_hom} with index $i-1$, one obtains the homogenized dynamics of vehicle $i-1$ as: 
\begin{align} \label{LTI_platoon_obs}
    \begin{bmatrix}
    \hspace{-0.05cm}\dot{v}_{i-1}(t)\hspace{-0.05cm}\\
    \hspace{-0.05cm}\dot{a}_{i-1}(t)\hspace{-0.05cm}\\
    \hspace{-0.05cm}\dot{u}_{bl_{i-1}}(t)\hspace{-0.05cm}
    \end{bmatrix} &\hspace{-0.1cm}=\hspace{-0.1cm} \underbrace{\begin{bmatrix}
    0 &  \hspace{-0.075cm}1\hspace{-0.075cm} & \hspace{-0.075cm}0\\
    0 & \hspace{-0.075cm}-\frac{1}{\tau_0}\hspace{-0.075cm} & \hspace{-0.075cm}\frac{1}{\tau_0}\\
    \hspace{-0.075cm}-K_{p_0}\hspace{-0.075cm}-\hspace{-0.075cm}\frac{K_{d_0}}{h}\hspace{-0.075cm} & \hspace{-0.075cm}-K_{d_0}\hspace{-0.075cm} & \hspace{-0.075cm}-\frac{1}{h}
    \end{bmatrix}}_{A_o}\hspace{-0.05cm}
    \underbrace{\begin{bmatrix}
    \hspace{-0.05cm}v_{i-1}(t)\hspace{-0.05cm}\\\hspace{-0.05cm}a_{i-1}(t)\hspace{-0.05cm}\\\hspace{-0.05cm}u_{bl_{i-1}}(t)\hspace{-0.05cm}
    \end{bmatrix}}_{x_{o_{i-1}}(t)} \vspace{-0.8cm} %\\
		\end{align}\vspace{-0.55cm}
		\begin{align}
		& %\hspace{-1.8cm}
		+ \underbrace{\begin{bmatrix}
    0\\
    0\\
    \frac{1}{h}
    \end{bmatrix}}_{B_o}
		\underbrace{\begin{bmatrix}
    K_{p_0}(q_{i-2}(t)-q_{i-1}(t))+K_{d_0}v_{i-2}(t)+u_{bl_{i-2}}(t)
    \end{bmatrix}}_{\nu_{o_{i-1}}(t)}, \nonumber
\end{align}
where subscript $o$ refers to the \emph{observer}. Among the components of $x_{o_{i-1}}$,  
%are not available to vehicle $i$: 
only $v_{i-1}$ is measurable to vehicle $i$ from on-board sensors (by adding the tachometer measurement $v_i$ to the radar measurement $v_{i-1}-v_{i}$): $a_{i-1}$ and $u_{bl_{i-1}}$ are not available during failures. The quantity $\nu_{o_{i-1}}$ in \eqref{LTI_platoon_obs} must be regarded as an \emph{unknown input} as it contains signals not available to vehicle $i$. 
Therefore, vehicle $i$ should observe (i.e. reconstruct) the state in \eqref{LTI_platoon_obs} using available measurements. Due to the unknown input term $\nu_{o_{i-1}}$ in \eqref{LTI_platoon_obs}, the type of observer we seek is an \emph{unknown input observer}. Specifically, we design a sliding-mode observer, in line with \cite{KALSI2010347}.

\subsubsection{Auxiliary variables} 
In the state $x_{o_{i-1}}$ in \eqref{LTI_platoon_obs}, the only measurement available to vehicle $i$ is
\begin{align} \label{LTI_platoon_obs_match}
v_{i-1}(t) = H_o x_{o_{i-1}}(t),
\end{align}
where $H_o = [1 \ \ 0 \ \ 0]$. However, the dynamics \eqref{LTI_platoon_obs} do not satisfy the {\color{black}observer matching condition needed for standard unknown input observer design \cite{KALSI2010347}, i.e. we have
\begin{align} \label{LTI_platoon_obs_match}
\rank(B_o) \neq \rank(B_o H_o).
\end{align}
This implies that the measurement $v_{i-1}$ in \eqref{LTI_platoon_obs_match} alone is not enough to reconstruct the missing signal $u_{bl_{i-1}}$.} %Some auxiliary outputs are needed: t
To this purpose, let us combine $H_o$ in \eqref{LTI_platoon_obs_match} with {\color{black}the auxiliary matrix} %, so as to define the auxiliary output matrix 
\begin{align} \label{LTI_platoon_obs_aux}
C_o = \begin{bmatrix}
    H_o  \\ H_o A_o \\ H_o A_o^2
    \end{bmatrix} = \begin{bmatrix}
    1 & 0 & 0  \\ 0 & 1 & 0  \\ 0 & -\frac{1}{\tau_0} & \frac{1}{\tau_0} 
    \end{bmatrix}.
\end{align}
%The main difference with respect to \cite{KALSI2010347} is that the physics interpretation of the virtual outputs allows a more direct design of the observer. In fact, the variables $C_o x_{o_{i-1}}$ defined by \eqref{LTI_platoon_obs_aux} are $v_{i-1}$, $a_{i-1}$, $j_{i-1}$, i.e. the velocity, acceleration and jerk of vehicle $i-1$. The chain-of-integrators structure of these outputs allows to design the high-gain observer 
Note that the variables $C_o x_{o_{i-1}}$ defined by \eqref{LTI_platoon_obs_aux} are $v_{i-1}$, $a_{i-1}$, $j_{i-1}$, i.e. the velocity, acceleration and jerk of vehicle $i-1$. The physical interpretation of these variables allows a more direct design than \cite{KALSI2010347}, namely, the chain-of-integrators structure of these variables allows to design the high-gain observer 
\begin{align} \label{LTI_platoon_highgain}
    \begin{bmatrix}
    \dot{\bar{v}}_{i-1}(t)\\
    \dot{\bar{a}}_{i-1}(t)\\
    \dot{\bar{j}}_{i-1}(t)
    \end{bmatrix} \hspace{-0.1cm}=\hspace{-0.1cm} \begin{bmatrix}
    0 & \hspace{-0.1cm}1\hspace{-0.1cm} & 0\\
    0 & \hspace{-0.1cm}0\hspace{-0.1cm} & 1\\
    0 & \hspace{-0.1cm}0\hspace{-0.1cm} & 0
    \end{bmatrix}\hspace{-0.1cm} \begin{bmatrix}
    \bar{v}_{i-1}(t)\\
    \bar{a}_{i-1}(t)\\
    \bar{j}_{i-1}(t)
    \end{bmatrix} \hspace{-0.1cm}+\hspace{-0.1cm} \begin{bmatrix}
    \frac{\alpha_0}{\epsilon}\\
    \frac{\alpha_1}{\epsilon^2}\\
    \frac{\alpha_2}{\epsilon^3}
    \end{bmatrix}\hspace{-0.1cm}(v_{i-1}(t)\hspace{-0.1cm}-\hspace{-0.1cm}\bar{v}_{i-1}(t)),
\end{align}
where the state comprises the reconstructed $v_{i-1}$, $a_{i-1}$, $j_{i-1}$; $\epsilon>0$ is a small design constant; $\alpha_0$, $\alpha_1$, $\alpha_2$ are selected to make the observer error dynamics stable, i.e. the polynomial $\lambda^3 + \alpha_0 \lambda^2+ \alpha_1 \lambda+ \alpha_2$ is Hurwitz. The state of the high-gain observer \eqref{LTI_platoon_highgain} is now used to reconstruct the missing $u_{bl_{i-1}}$. 

\subsubsection{Observer for missing signal $u_{bl_{i-1}}$} The unknown input observer  is selected as %in the sliding-mode form 
\begin{align} \label{LTI_platoon_unkinput}
\dot{\hat{x}}_{o_{i-1}}(t) &= A_o \hat{x}_{o_{i-1}}(t) + B_o E_o(t) \nonumber\\
&\qquad+ L_o (\bar{y}_{{i-1}}(t) - \hat{y}_{i-1}(t)),
\end{align}
where $\hat{y}_{i-1}(t)=C_o \hat{x}_{o_{i-1}}(t)$ are the reconstructed {\color{black}auxiliary variables}, $L_o$ is such that $(A_o-L_oC_o)$ is Hurwitz, and 
%\vspace{-0.2cm}
\begin{align} \label{LTI_platoon_unkinput3}
   \hat{x}_{o_{i-1}}(t)\hspace{-0.08cm} = \hspace{-0.08cm}\begin{bmatrix}
    \hat{v}_{i-1}(t)\\ \hat{a}_{i-1}(t)\\ \hat{u}_{bl_{i-1}}(t)
    \hspace{-0.08cm}\end{bmatrix}, \  \bar{y}_{{i-1}}(t) \hspace{-0.08cm}= \hspace{-0.08cm}\begin{bmatrix}
    v_{i-1}(t)\\
    S_a \sat(\frac{\bar{a}_{i-1}(t)}{S_a})\\
    S_j \sat(\frac{\bar{j}_{i-1}(t)}{S_j})
    \hspace{-0.08cm}\end{bmatrix},
\end{align}
\vspace{-0.65cm}
\begin{align} \label{LTI_platoon_unkinput2}
 & E_o(t) =  \\
 &\ \ \left\{\begin{array}{ll}  
\eta \frac{F_o(\bar{y}_{{i-1}}(t)-\hat{y}_{i-1}(t))}{\left\|\bar{y}_{{i-1}}(t)-\hat{y}_{i-1}(t)\right\|} & {\rm if} \;  F_o (\bar{y}_{{i-1}}(t)-\hat{y}_{i-1}(t)) \neq 0 \\
 0 & {\rm otherwise,} 
\end{array}\right. \nonumber
\end{align}
with $\sat(\cdot)$ being a standard saturation function in the range $\pm 1$. In \eqref{LTI_platoon_unkinput3}-\eqref{LTI_platoon_unkinput2}, $\eta>0$ is a sliding mode gain, $S_a$, $S_j>0$ are bounds for maximum acceleration and jerk in the platoon, and  $F_o$ is designed such that
\begin{align} \label{LTI_platoon_unkinput4}
   & (A_o-L_oC_o)^T P_o + P_o (A_o-L_oC_o) = -2 Q_o, \nonumber \\
	& F_o C_o = B_o^T P_o, \quad P_o, Q_o >0,
\end{align}
for some $Q_o>0$. {\color{black}Using %the fact that $P_o \in \mathbb{R}^{3 \times 3}$, and using 
 the notation $p_{o_{ij}}$ for the entry $(i,j)$ of $P_o\in \mathbb{R}^{3 \times 3}$, it is obtained from direct calculation of $F_o C_o = B_o^T P_o$ in \eqref{LTI_platoon_unkinput4} so that $F_o \in \mathbb{R}^{1 \times 3}$}  takes the form
\begin{align} \label{LTI_platoon_unkinput5}
F_o  = \begin{bmatrix}
    \frac{p_{o_{13}}}{h} & \frac{p_{o_{32}}+p_{o_{33}}}{h} & \frac{p_{o_{33}}}{h}\tau_0
    \end{bmatrix},
\end{align}
to be used in the observer \eqref{LTI_platoon_unkinput2}. The following result holds.

\begin{proposition}[Unknown input observer \cite{KALSI2010347}]
For the vehicle dynamics \eqref{LTI_platoon_obs}, the observer  \eqref{LTI_platoon_highgain}-\eqref{LTI_platoon_unkinput4} guarantees that there exists a constant $\epsilon^* \in (0, 1)$ such that, if $\epsilon \in (0, \epsilon^*)$ and $\eta \geq \left|\nu_{o_{i-1}}(t)\right|$, %$\forall t$, then 
the state estimation error $x_{o_{i-1}}(t)-\hat{x}_{o_{i-1}}(t)$ is uniformly ultimately
bounded by a function $\kappa(\epsilon)$ that satisfies $\kappa(\epsilon)\rightarrow 0$ when $\epsilon\rightarrow 0$.
\end{proposition}

Proposition 4 suggests that the control input is finally
\begin{equation}\label{cont_hom_obser}
h\dot{u}_{bl_i}(t) = -u_{bl_i}(t) + K_{p_0} e_i(t) + K_{d_0} \dot{e}_i(t) + \hat{u}_{bl_{i-1}}(t),
\end{equation}
where the missing $u_{bl_{i-1}}(t)$ has been replaced by $\hat{u}_{bl_{i-1}}(t)$. 

%\begin{algorithm}
%\floatname{algorithm}{Resilient Strategy}
%\caption{(against communication failures)}
%\label{array-sum}
%\begin{algorithmic}[1]
%\Require a-priori unknown ($K_{p_0}$, $K_{d_0}$, $\tau_0$) from convergence of \eqref{cons_hom}; Lyapunov equation matrix $Q_o$ and observer gain $L_o$; \newline %$S_a$, $S_j$ (maximum acceleration and jerk), 
%sliding gain $\eta$ and high-gain observer gains $\epsilon$, $\alpha_0$, $\alpha_1$, $\alpha_2$. 
%\Ensure control \eqref{cont_hom_obser}, with the reconstructed $\hat{u}_{bl_{i-1}}(t)$. % in place of the missing variable $u_{bl_{i-1}}(t)$. 
%\vspace{0.05cm}
%\State Form a copy of the system matrices $(A_o,B_o)$ as in \eqref{LTI_platoon_obs} (homogenized preceding vehicle matrices);
%\State By means of the high-gain observer  \eqref{LTI_platoon_highgain}, reconstruct $v_{i-1}$, $a_{i-1}$, $j_{i-1}$ %, i.e. the velocity, acceleration and jerk of vehicle $i-1$, 
%as $\hat{v}_{i-1}$, $\hat{a}_{i-1}$, $\hat{j}_{i-1}$; 
%\State By means of the unknown input observer \eqref{LTI_platoon_unkinput}, reconstruct the missing $u_{bl_{i-1}}$ as $\hat{u}_{bl_{i-1}}$, and use it in \eqref{cont_hom_obser}.
%\end{algorithmic}\vspace{-0.1cm}
%\end{algorithm}%\vspace{-0.3cm}

\begin{remark}[\emph{Homogeneous a-priori unknown observer}]
Each vehicle $i$ does not use the heterogeneous parameters $\tau_{i-1}$, $K_{p_{i-1}}$, $K_{d_{i-1}}$ of the preceding vehicle $i-1$ to design the observer, but homogenized preceding vehicle dynamics defined by $(A_o,B_o)$ in \eqref{LTI_platoon_obs}. As $(A_o,B_o)$ stems from the common group behavior originating from $K_{p_0}$, $K_{d_0}$, $\tau_0$, each vehicle ends up designing the same observer, regarded as a \emph{common group observer}.
\end{remark}

The proposed resilient approach can be sketched along the aforementioned three steps as:%\newpage%\vspace{-0.3cm}

\begin{algorithm}
\floatname{algorithm}{Resilient Strategy}
\caption{(against communication failures)}
\label{array-sum}
\begin{algorithmic}[1]
\Require a-priori unknown $K_{p_0}$, $K_{d_0}$, $\tau_0$ from convergence of \eqref{cons_hom}; Lyapunov equation matrix $Q_o$ and observer gain $L_o$; \newline %$S_a$, $S_j$ (maximum acceleration and jerk), 
sliding gain $\eta$ and high-gain observer gains $\epsilon$, $\alpha_0$, $\alpha_1$, $\alpha_2$. 
\Ensure control \eqref{cont_hom_obser}, with the reconstructed $\hat{u}_{bl_{i-1}}(t)$. % in place of the missing variable $u_{bl_{i-1}}(t)$. 
\vspace{0.05cm}
\State Offline: form a copy of the homogenized preceding vehicle matrices $(A_o,B_o)$ as in \eqref{LTI_platoon_obs};
\State Online, at time $t$: %by means of the high-gain observer  \eqref{LTI_platoon_highgain}, 
reconstruct $v_{i-1}(t)$, $a_{i-1}(t)$, $j_{i-1}(t)$ %, i.e. the velocity, acceleration and jerk of vehicle $i-1$, 
as $\bar{v}_{i-1}(t)$, $\bar{a}_{i-1}(t)$, $\bar{j}_{i-1}(t)$ via the high-gain observer  \eqref{LTI_platoon_highgain}; 
\State Online, at time $t$: %by means of the unknown input observer \eqref{LTI_platoon_unkinput}, 
reconstruct the missing $u_{bl_{i-1}}(t)$ as $\hat{u}_{bl_{i-1}}(t)$ via the %unknown input 
observer \eqref{LTI_platoon_unkinput}, and use $\hat{u}_{bl_{i-1}}(t)$ in \eqref{cont_hom_obser}.
\end{algorithmic}\vspace{-0.1cm}
\end{algorithm}%\vspace{-0.3cm}

\section{Numerical experiments} \label{Chap_simu}

%The numerical validations are organized as follows. Sec.~V-A validates the benefits of self-organization, while Sec.~V-B validates the proposed resilient designs based on self-organization. % solve the issues caused by acceleration limits and communication failures. 
%{\color{black}Secs.~V-C and V-D provide implementations of the proposed ideas in CommonRoad \cite{7995802} and CARLA \cite{xu2021opencda}: the former tests full-range operation by means of the safety architecture in Sec.~II-C; the latter tests the proposed methods with the vehicle/sensor models available in CARLA.} Sec.~V-E considers cut-in and merging scenarios, to validate the effectiveness %of the proposed ideas 
%in more complex traffic scenarios. {\color{black}Numerical validations are carried out in presence of zero-mean Gaussian noise to simulate radar measurement noise (variance 0.025 m$^2$/s$^2$), tachometer measurement noise (variance 0.25 m$^2$/s$^2$) and accelerometer measurement noise (variance 0.1 m$^4$/s$^4$). These noises prevent ideal asymptotic convergence of the errors, but allow to validate the proposed strategies in a more realistic setting and to test sensitivity of controllers and observers to noise}. 
The numerical experiments are organized as follows. Sec.~V-A evaluates the proposed resilient designs based on self-organization. % solve the issues caused by acceleration limits and communication failures. 
{\color{black}Secs.~V-B and V-C provide implementations of the proposed ideas in CommonRoad \cite{7995802} and CARLA \cite{xu2021opencda}: the former tests full-range operation by means of the safety architecture in Sec.~II-C; the latter tests the proposed methods with the vehicle/sensor models available in CARLA.} {\color{black}Sec.~V-D considers cut-in and merging scenarios, to validate the effectiveness %of the proposed ideas 
in more complex traffic scenarios.} {\color{black}Numerical experiments are carried out in presence of zero-mean Gaussian noise to simulate radar measurement noise (variance 0.025 m$^2$/s$^2$), tachometer measurement noise (variance 0.25 m$^2$/s$^2$) and accelerometer measurement noise (variance 0.1 m$^4$/s$^4$). These noises prevent ideal asymptotic convergence of the errors, but evaluate the proposed strategies in a more realistic setting and test the sensitivity of controllers and observers to noise}.

%\subsection{{\color{black}Validation of self-organization}}
%
%Let us consider a CACC-equipped platoon of one leading and five following vehicles, with headway $h=0.7$ and heterogeneous parameters in as Tab.~\ref{param}. First, let us keep the proposed self-organizing strategy \eqref{cons_hom}, \eqref{gain_hom2} off ($u_{hm_i}=\xi_{hm_i}=0$), i.e. every vehicle keeps its heterogeneous dynamics. We let the leader first accelerate, then decelerate. The results are in Figs.~\ref{vel_het} and \ref{error_het}: %Although Fig.~\ref{vel_het} looks the same as Fig.~\ref{vel_ideal}, t
%although all vehicles manage to follow the velocity and acceleration of the leader, the %this is not true: 
%absence of self-organization results in non-converging spacing errors due to heterogeneity. %, in contrast to the converging spacing errors in 
%Recall that the distance to the leading vehicle is velocity-dependent, %increases with increasing velocity and decreases with decreasing velocity 
%as a result of having $hv_i$ in the constant time headway policy \eqref{error}. 

\begin{table}[h]
\captionsetup{width=.75\textwidth}
\caption{Heterogeneous parameters of the platoon.}
\label{param}
\centering
%\begin{center}
{\footnotesize
\begin{tabular}{lllllll}
\toprule
 & 1 & 2 & 3 & 4 & 5 & 6\\
\midrule
$\tau_i $ & 0.10 & 0.20 & 0.05 & 0.30 & 0.15 & 0.075\\
$K_{p_i} $ & 0.20 & 0.10 & 0.40 & 0.067 & 0.133 & 0.267\\
$K_{d_i} $ & 0.70 & 0.35 & 1.40 & 0.23 & 0.467 & 0.933\\
\bottomrule
\end{tabular}
}
%\end{center}
\end{table}

%For the same leader acceleration/deceleration scenario, let us now switch the proposed self-organizing strategy \eqref{cons_hom}, \eqref{gain_hom2} on, with $\mu_p =\mu_d = \mu_\tau = 0.2$. %The leader first accelerates, then decelerates, while the self-organizing strategy \eqref{cons_hom}, \eqref{gain_hom2} is switched on. % (i.e. with convergence to common engine constant, common proportional and derivative gains). 
%Fig.~\ref{vel_ideal} shows that all vehicles manage to follow the velocity and acceleration of the leader, the effect of noise on acceleration being smaller than Fig.~\ref{vel_het}. Meanwhile, Fig.~\ref{error_ideal} shows that the spacing error is kept closer to zero as compared to Fig.~\ref{error_het}.  

%To show the benefit of the proposed self-organizing strategy \eqref{cons_hom}, \eqref{gain_hom2}, we switched off self-organization ($u_{hm_i}=\xi_{hm_i}=0$), i.e. every vehicle keeps its heterogeneous dynamics. For the same leader acceleration/deceleration scenario, the results are in Figs.~\ref{vel_het} and \ref{error_het}. %Although Fig.~\ref{vel_het} looks the same as Fig.~\ref{vel_ideal}, t
%The comparison between Figs.~\ref{error_het} and \ref{error_ideal} reveals that the %this is not true: 
%absence of self-organization results in non-converging spacing errors. %, in contrast to the converging spacing errors in Fig.~\ref{error_ideal}. 

\subsection{Resilience to {\color{black}acceleration} limits and communication failures}

Let us consider a CACC-equipped platoon of one leading and five following vehicles, with headway $h=0.7$ and heterogeneous parameters as in Tab.~\ref{param}. To clarify the adverse effect of {\color{black}acceleration} limits towards formation keeping, consider the heterogeneous limits in Tab.~\ref{trueP2}. %As in Sec.~V-A, the leader first accelerates, then decelerates, but this time we chop the acceleration within $\pm 0.425$ [m/s$^2$] due to the leader acceleration limits. 
We let the leader accelerate, then decelerate, within its acceleration limits $\pm 0.425$ [m/s$^2$]. % due to its acceleration limits.

Fig.~\ref{vel_sat} shows that several followers with even tighter acceleration limits than the leader are not able to follow the leader velocity and acceleration.
%  than the limits of the leader acceleration limits of the followers %(often tighter than the limits of the leader) 
%prevent from following the leader velocity and acceleration. 
As a result, Fig.~\ref{error_sat} shows that the platoon first loses cohesiveness during the acceleration phase, and then a collision happens during the deceleration phase. This can be seen from the relative distance crossing zero at $t \approx 110$s. The collision is caused by the fact that most vehicles cannot follow the deceleration of the leader due to their deceleration limits. 

\begin{table}[h]
\captionsetup{width=.75\textwidth}
\caption{\text{Heterogeneous limits of the platoon.}}
\label{trueP2}
\centering
%\begin{center}
{\footnotesize
\begin{tabular}{lllllll}
\toprule
  & 1 & 2 & 3 & 4 & 5 & 6\\
\midrule
$a_{max_i}$ & \;0.425 & \;0.35 & \;0.375 & \;0.40 & \;0.325 & \;0.45\\
%\hline
$a_{min_i}$ & -0.425 & -0.35 & -0.375 & -0.40 & -0.325 & -0.45\\
\bottomrule
\end{tabular}
}
\end{table}

\begin{figure}[t]
      \centering
			\begin{subfigure}{0.5\textwidth}
      \centering
      \includegraphics[width=0.99\textwidth]{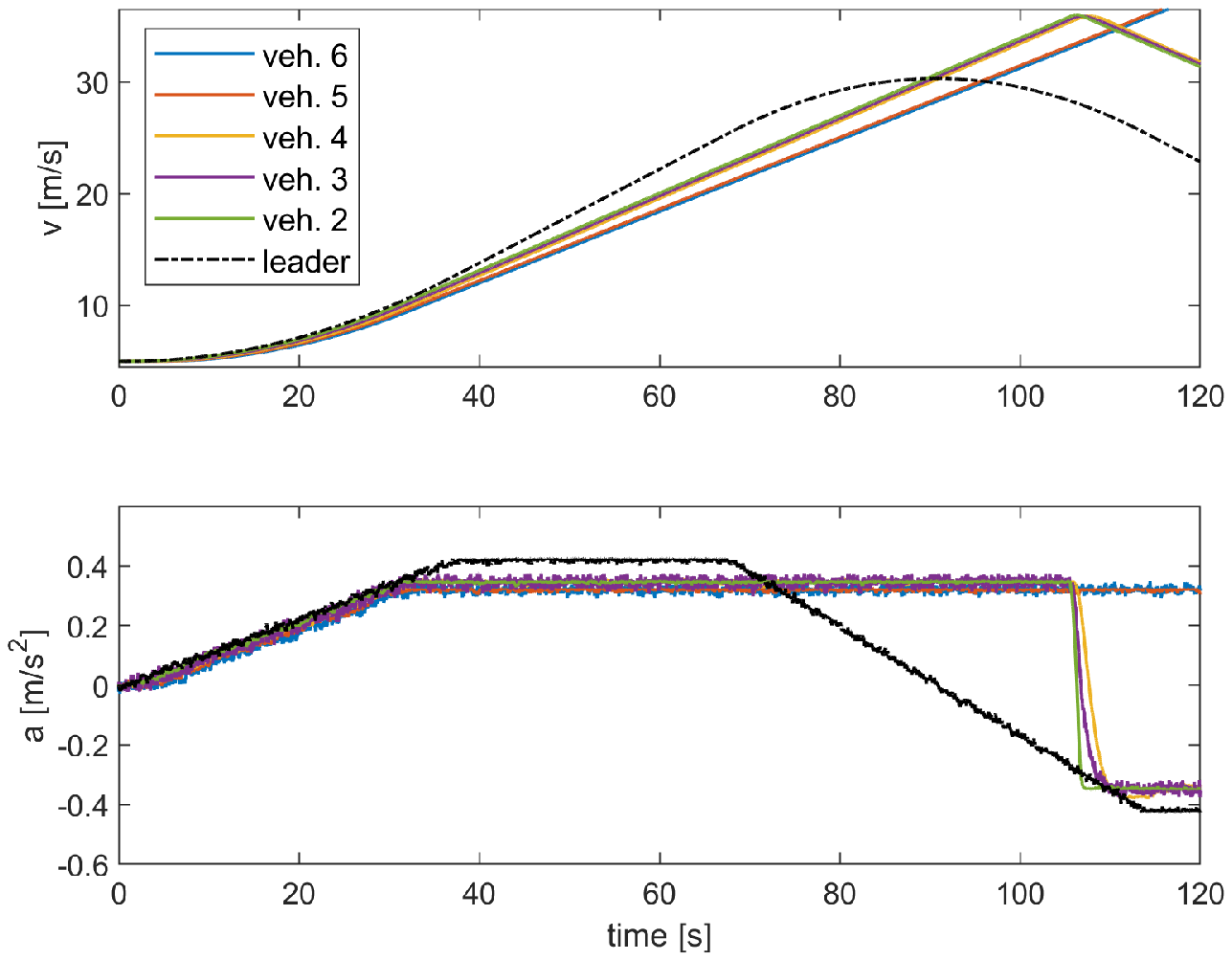}
  \caption{Velocity (upper plot) and acceleration (lower plot) for all \newline vehicles. %The leading vehicle is shown in dash-dotted line. 
	The platoon cannot follow the leader.}
  \label{vel_sat}
\end{subfigure}
    ~ %add desired spacing between images, e. g. ~, \quad, \qquad, \hfill etc. 
\begin{subfigure}{0.5\textwidth}
  \centering
  % include second image
      \includegraphics[width=0.99\textwidth]{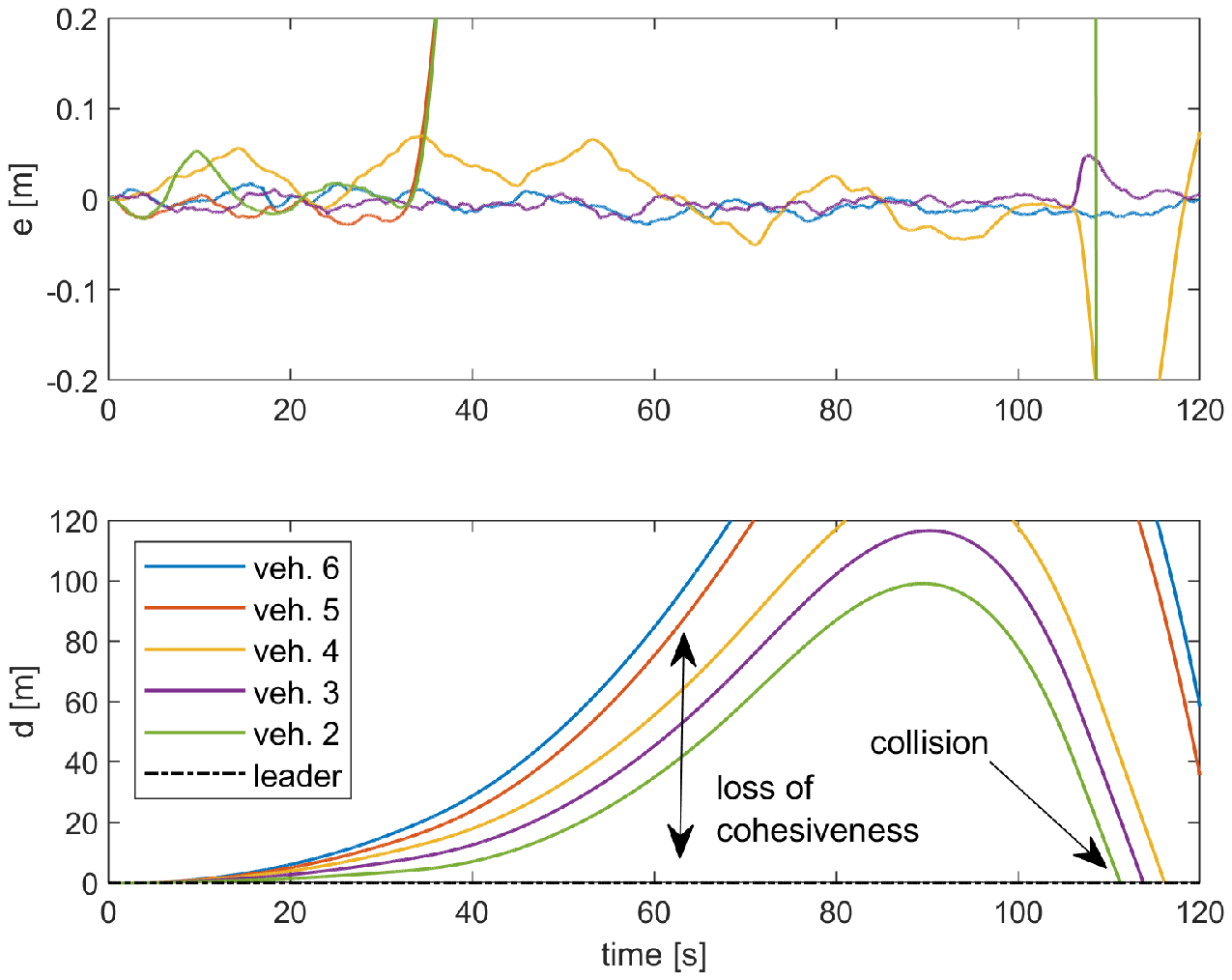}
  \caption{Spacing error (upper plot) and distance to leading vehicle \newline (lower plot). A collision happens in the platoon.}
  \label{error_sat}
\end{subfigure}
\caption{{\color{black}Acceleration} limits, without resilient strategy.}
\label{TOT_sat}
\end{figure}

\begin{figure}[t]
      \centering
			\begin{subfigure}{0.5\textwidth}
      \centering
  \includegraphics[width=0.99\textwidth]{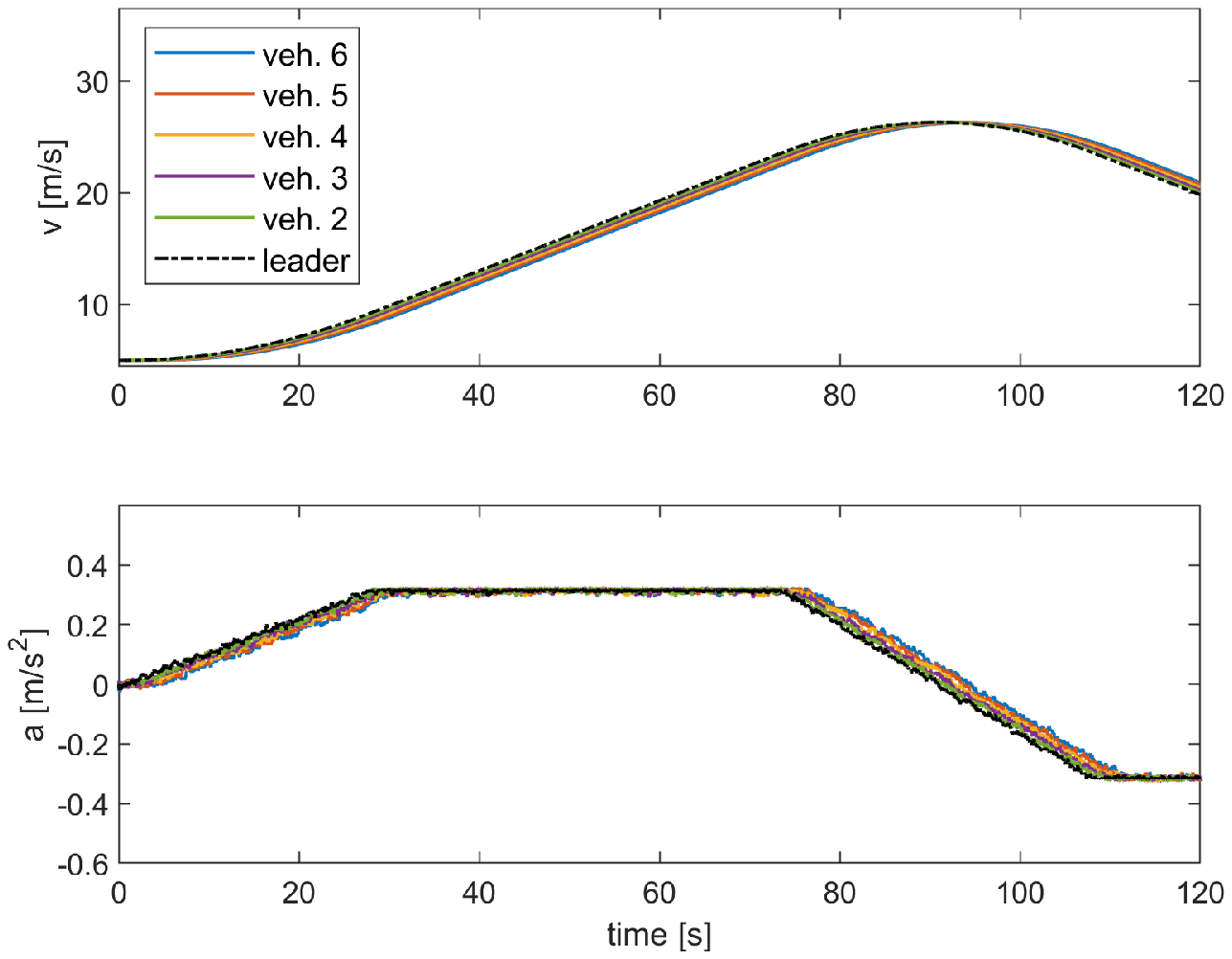}
  \caption{Velocity (upper plot) and acceleration (lower plot) for all \newline vehicles. The platoon follows the leader.} % at the price of reduced performance.}
  \label{vel_sat_prop}
\end{subfigure}
    ~ %add desired spacing between images, e. g. ~, \quad, \qquad, \hfill etc. 
\begin{subfigure}{0.5\textwidth}
  \centering
  % include second image
  \includegraphics[width=0.99\textwidth]{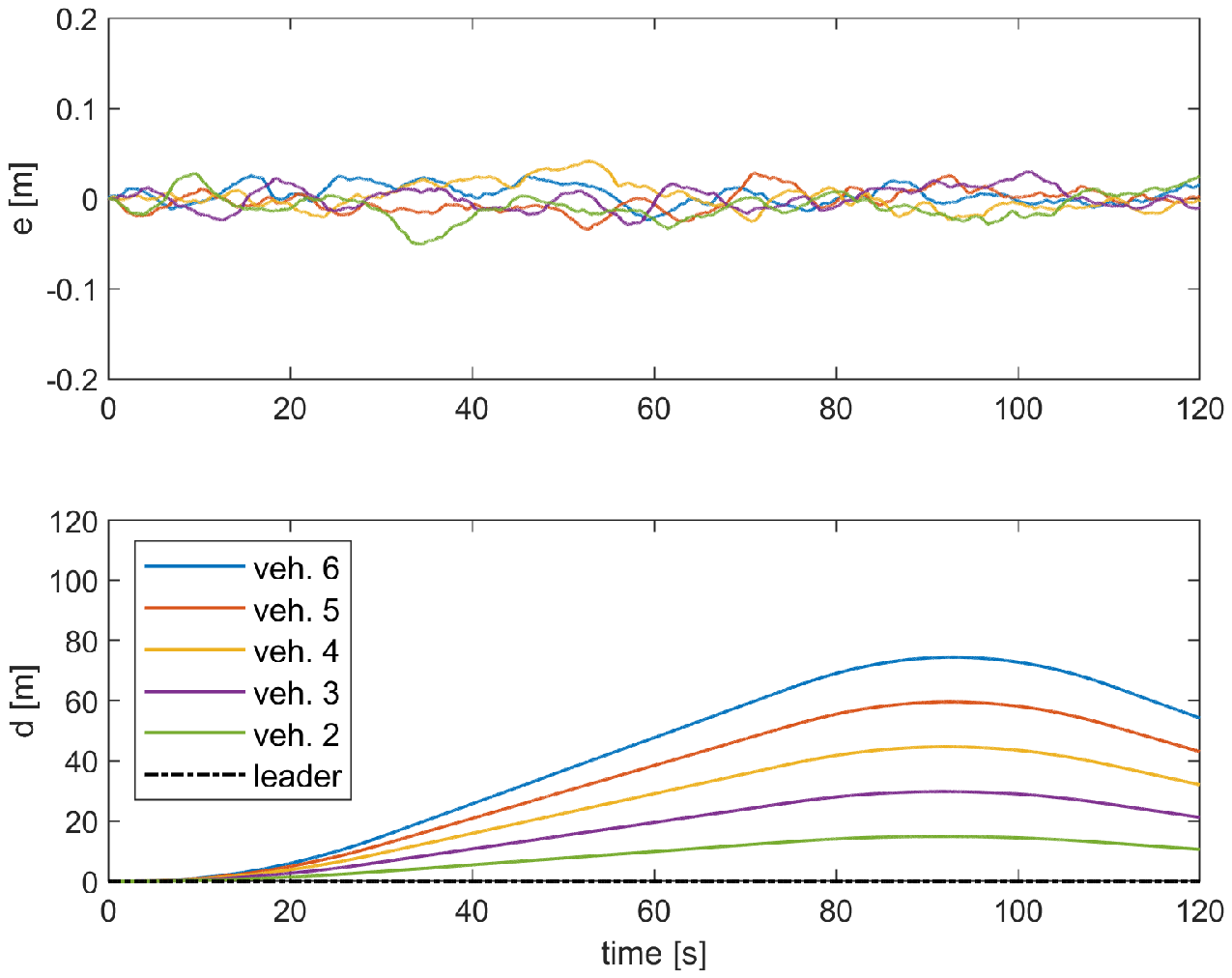}
  \caption{Spacing error (upper plot) and distance to leading vehicle (lower plot). The spacing error is regulated despite the {\color{black}acceleration} limits.}
  \label{error_sat_prop}
\end{subfigure}
\caption{{\color{black}Acceleration} limits, proposed resilient strategy.}
\label{TOT_sat_prop}
\end{figure}

\begin{figure}[t]
      \centering
			\begin{subfigure}{0.5\textwidth}
      \centering
  \includegraphics[width=0.99\textwidth]{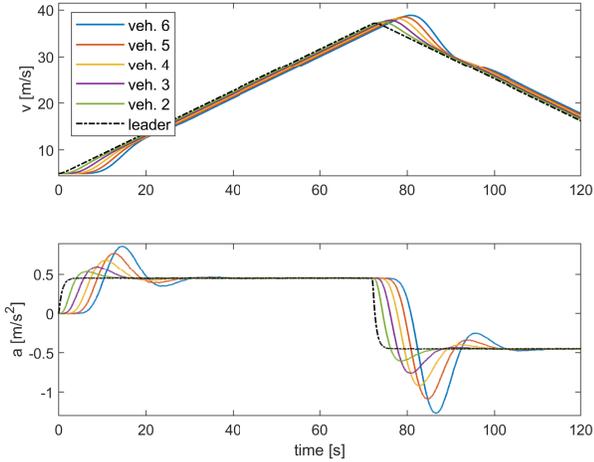}
  \caption{Velocity (upper plot) and acceleration (lower plot) for all \newline vehicles. Note the large velocity and acceleration transients.}% during The platoon follows the leader behavior with large transients.}
  \label{vel_comm}
\end{subfigure}
    ~ %add desired spacing between images, e. g. ~, \quad, \qquad, \hfill etc. 
\begin{subfigure}{0.5\textwidth}
  \centering
  % include second image
  \includegraphics[width=0.99\textwidth]{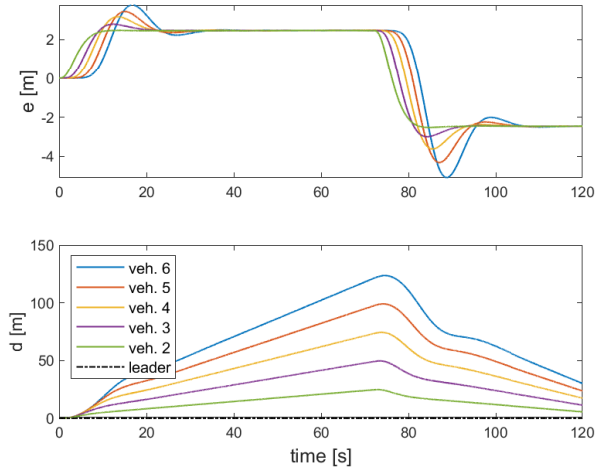}
  \caption{Spacing error (upper plot) and distance to leading vehicle \newline (lower plot). The spacing error is not regulated close to zero.}
  \label{error_comm}
\end{subfigure}
\caption{Communication failures, without resilient strategy.}
\label{TOT_comm}
\end{figure}

\begin{figure}[t]
      \centering
			\begin{subfigure}{0.5\textwidth}
      \centering
  \includegraphics[width=0.99\textwidth]{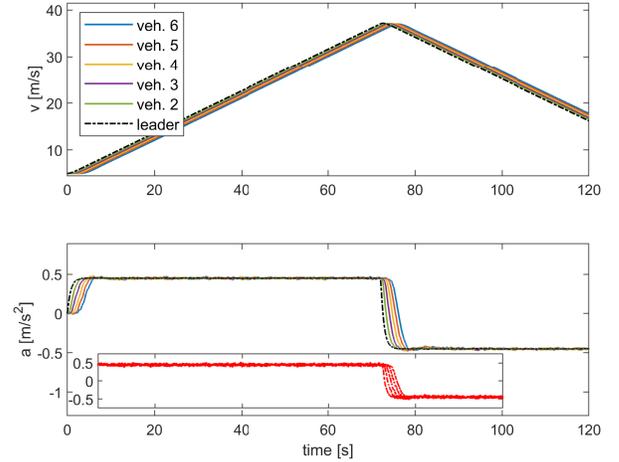}
  \caption{Velocity (upper plot) and acceleration (lower plot) for all \newline vehicles, with observer estimates in red in the small box.} %(red). Close-to-ideal performance is recovered.}
  \label{vel_comm_prop}
\end{subfigure}
    ~ %add desired spacing between images, e. g. ~, \quad, \qquad, \hfill etc. 
\begin{subfigure}{0.5\textwidth}
  \centering
  % include second image
  \includegraphics[width=0.99\textwidth]{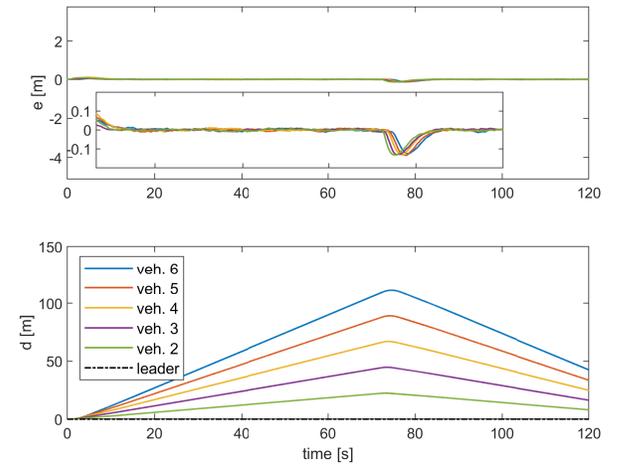}
  \caption{Spacing error (upper plot) and distance to leading vehicle \newline (lower plot). Close-to-ideal performance is recovered.} %After the transient phases, spacing errors converge to zero.}
  \label{error_comm_prop}
\end{subfigure}
\caption{Communication failures, proposed resilient strategy.}
\label{TOT_comm_prop}
\end{figure}

%With the same acceleration limits of Tab.~\ref{trueP2}, Fig.~\ref{vel_sat_prop} demonstrates the effectiveness of the proposed resilient strategy against acceleration limits:  
%all vehicles manage to follow the velocity and acceleration of the leader. Of course, this is done at the expense of reduced performance. Indeed,  the platoon %the leader `limits' its behavior and 
%reaches lower velocity (25m/s in Fig.~\ref{vel_sat_prop} compared to 32m/s in Fig.~\ref{vel_ideal}), but this is not a surprise: it was already shown in \cite{751767,1383798} that reducing performance is necessary to keep the formation, since it avoids that the vehicles hit their acceleration limits. Fig.~\ref{error_sat_prop} shows that the spacing error is regulated in a way similar to Fig.~\ref{error_ideal}, despite no acceleration limits were present in the latter. %) is recovered. % with the non-convergence in Fig.~\ref{error_sat} (saturated case) and with the convergence in Fig.~\ref{error_sat_prop} (proposed). 
%Therefore, it is confirmed that self-organization to a constrained group model promotes resilience to acceleration limits, %: the platoon achieves a homogeneous behavior 
%since even the vehicle with the worst acceleration limits can follow the behavior of the leader.
With the same acceleration limits of Tab.~\ref{trueP2}, Fig.~\ref{vel_sat_prop} demonstrates the effectiveness of the proposed resilient strategy against acceleration limits:  
all vehicles manage to follow the velocity and acceleration of the leader. Not surprizingly, this is done at the expense of lower velocity: it was already shown in \cite{751767,1383798} that reducing performance is necessary to keep the formation, since it avoids that the vehicles hit their acceleration limits. Fig.~\ref{error_sat_prop} shows that the spacing error is regulated close to zero (perfect regulation to zero is not possible due to noises). %, despite no acceleration limits were present in the latter. %) is recovered. % with the non-convergence in Fig.~\ref{error_sat} (saturated case) and with the convergence in Fig.~\ref{error_sat_prop} (proposed). 
Therefore, it is confirmed that self-organization to a constrained group model promotes resilience to acceleration limits, %: the platoon achieves a homogeneous behavior 
since even the vehicle with the worst acceleration limits can follow the behavior of the leader.

To illustrate the adverse effects of losing communication, let us remove acceleration limits for simplicity. 
Also, to make the effect of losing $u_{bl_{i-1}}$ as in \eqref{cont_hom_nocomm} more evident, we make the leader accelerate/decelerate with sharp piecewise constant phases. Fig.~\ref{vel_comm} shows that communication failures result in poor vehicle-following capabilities, with transients in acceleration and velocity, and with spacing errors that cannot converge and grow in the order of 5 meters (Fig.~\ref{error_comm}).

\begin{figure*}[t]
      \centering
\begin{subfigure}{0.48\textwidth}
      \centering
      \includegraphics[width=0.99\textwidth]{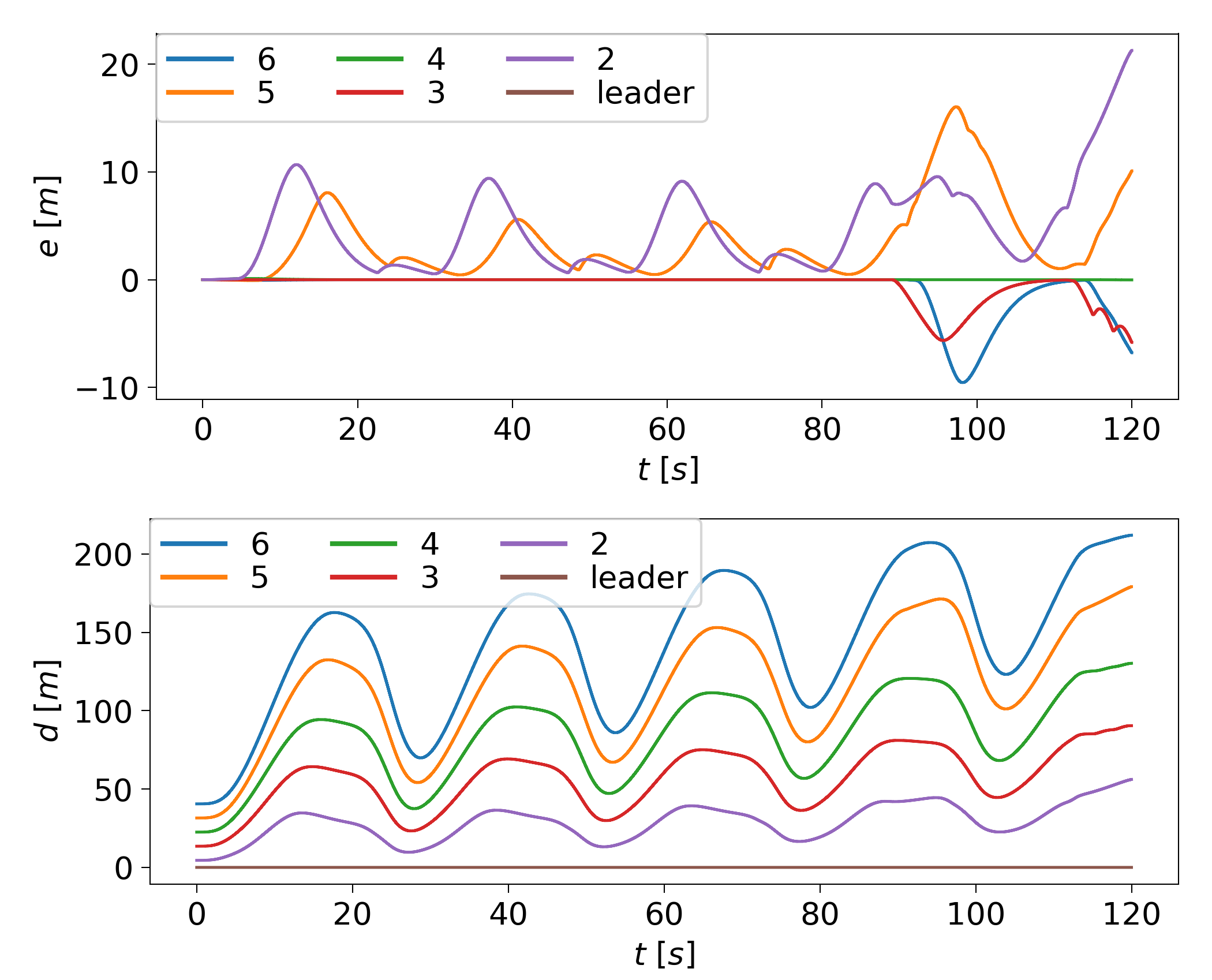}
  \caption{Spacing error and distance to leader without self-organization.}
  \label{HET_OFF}
\end{subfigure}
    ~ %add desired spacing between images, e. g. ~, \quad, \qquad, \hfill etc. 
\begin{subfigure}{0.48\textwidth}
  \centering
      \includegraphics[width=0.99\textwidth]{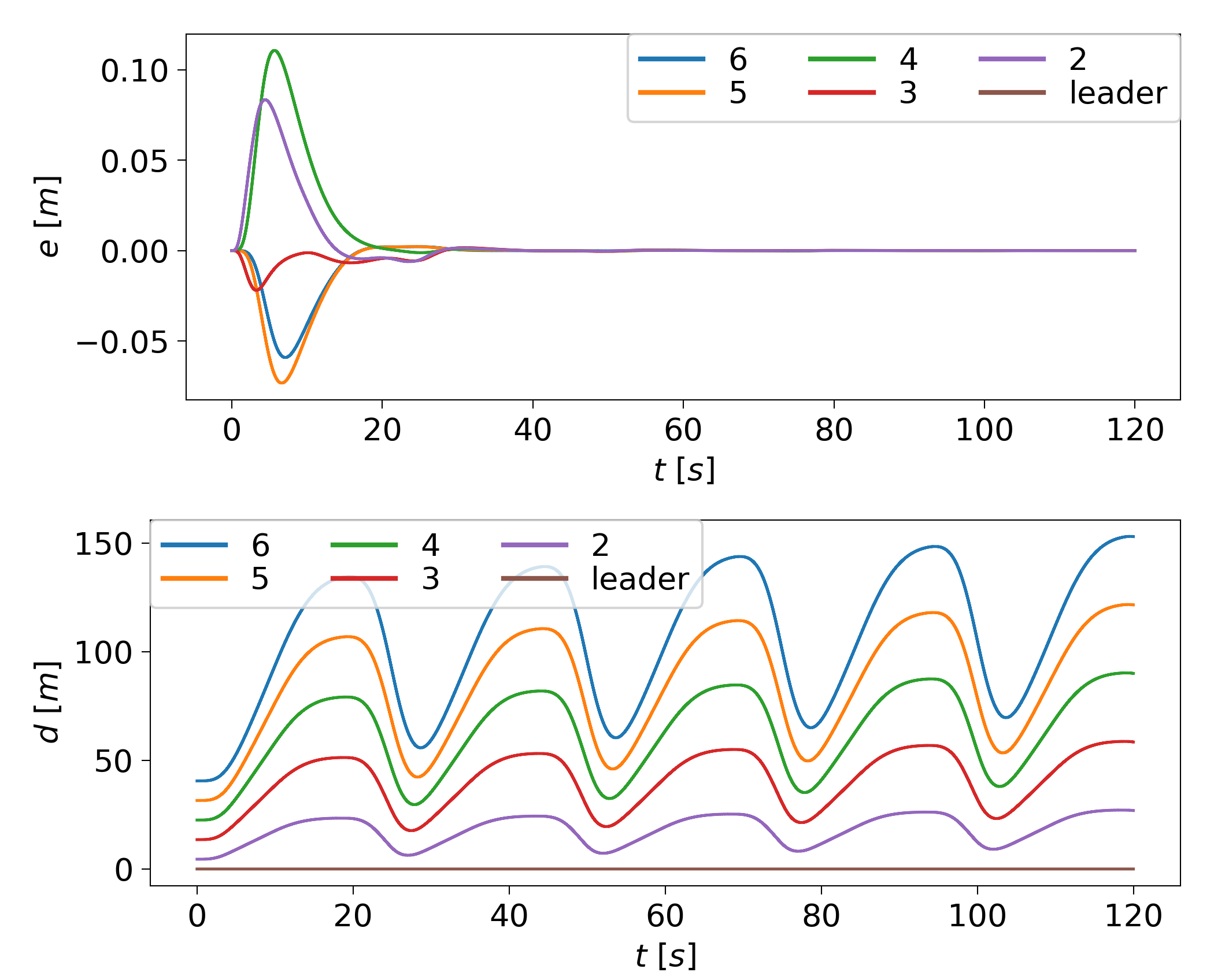}
  \caption{Spacing error and distance to leader with self-organization.}
  \label{HET_ON}
\end{subfigure}
\begin{subfigure}{0.48\textwidth}
  \centering
      \includegraphics[width=0.99\textwidth]{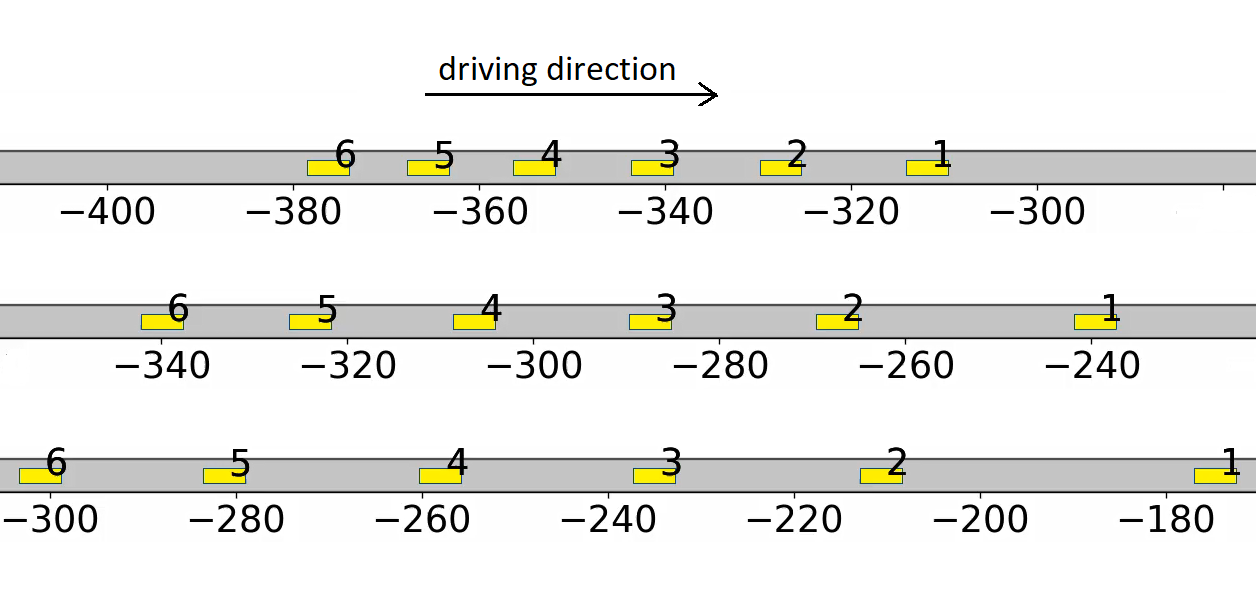}
  \caption{{\color{black}Snapshots from} CommonRoad without self-organization.}\vspace{-0.15cm}
  \label{HET_OFF_video}
	\end{subfigure}
    ~ %add desired spacing between images, e. g. ~, \quad, \qquad, \hfill etc. 
\begin{subfigure}{0.48\textwidth}
  \centering
      \includegraphics[width=0.99\textwidth]{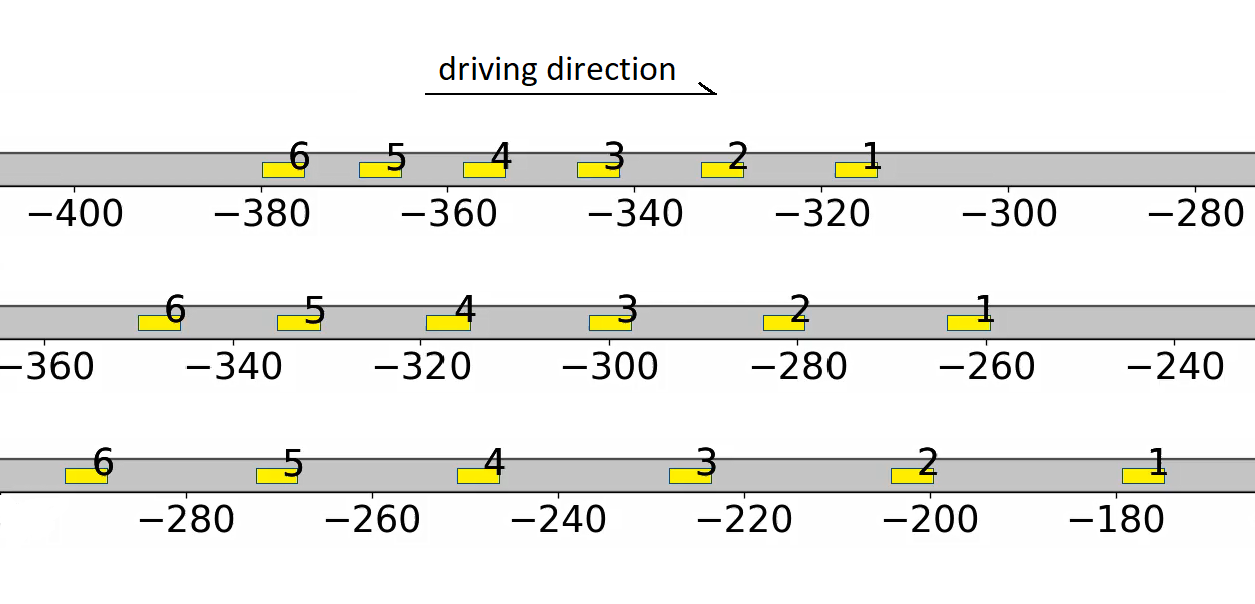}
  \caption{{\color{black}Snapshots from} CommonRoad with self-organization.}\vspace{-0.15cm}
  \label{HET_ON_video}
	\end{subfigure}
\caption{{\color{black}Spacing error, distance to leading vehicle, and snapshots from CommonRoad, with and without self-organizing strategy, in the presence of heterogeneous {\color{black}acceleration} limits. Our provably-correct safety layer is included. Self-organization allows the spacing errors to converge to zero.}}
\label{TOT_HOM_HET}
\end{figure*}

\begin{figure*}[t]
      \centering
\begin{subfigure}{0.48\textwidth}
      \centering
      \includegraphics[width=0.99\textwidth]{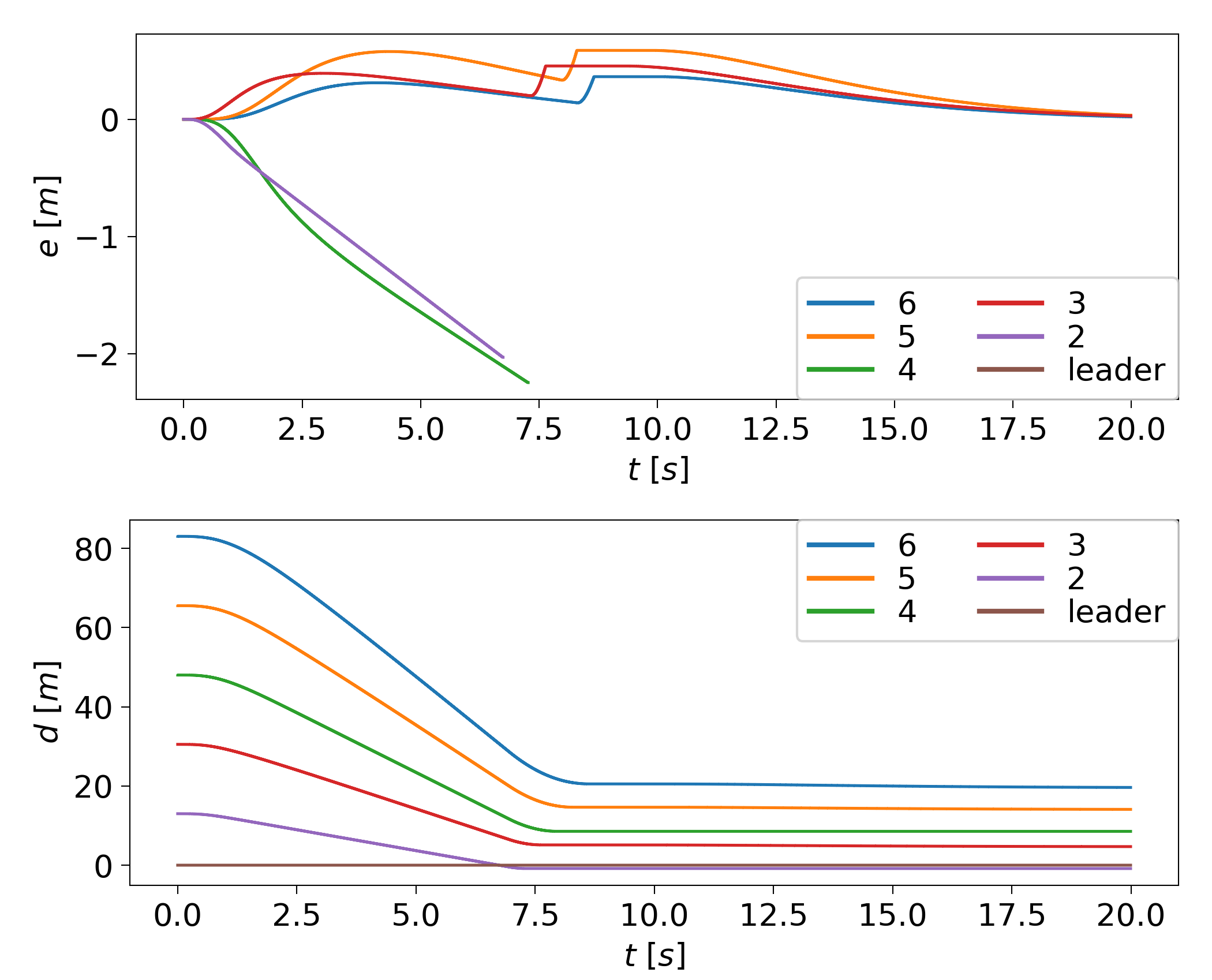}
  \caption{Emergency braking without safety layer.}
  \label{CRASH_OFF}
\end{subfigure}
    ~ %add desired spacing between images, e. g. ~, \quad, \qquad, \hfill etc. 
\begin{subfigure}{0.48\textwidth}
  \centering
      \includegraphics[width=0.99\textwidth]{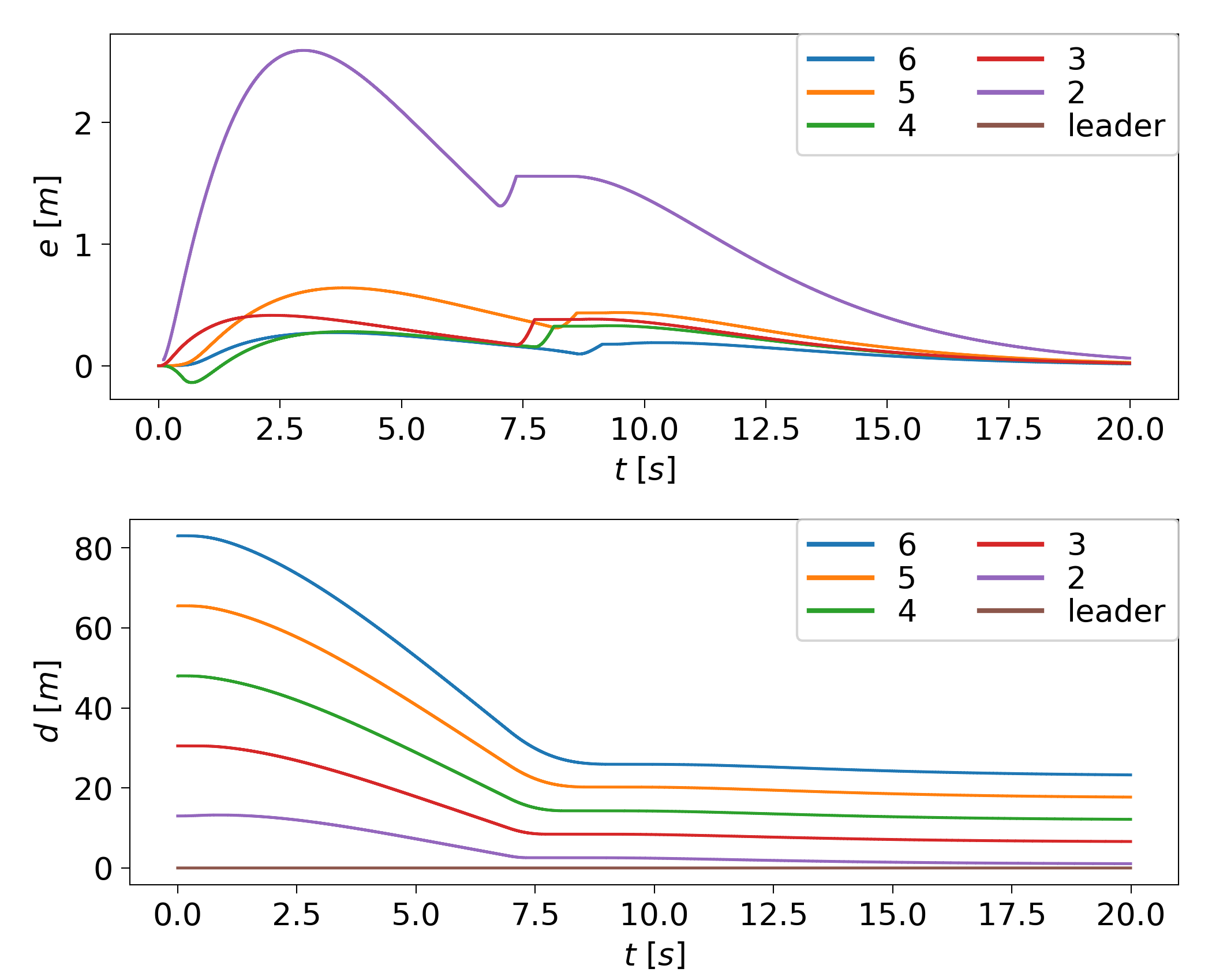}
  \caption{Emergency braking with safety layer.}
  \label{CRASH_ON}
	\end{subfigure}
        ~ %add desired spacing between images, e. g. ~, \quad, \qquad, \hfill etc. 
\begin{subfigure}{0.48\textwidth}
  \centering
      \includegraphics[width=0.99\textwidth,trim=0cm 0cm 0cm 0.2cm,clip]{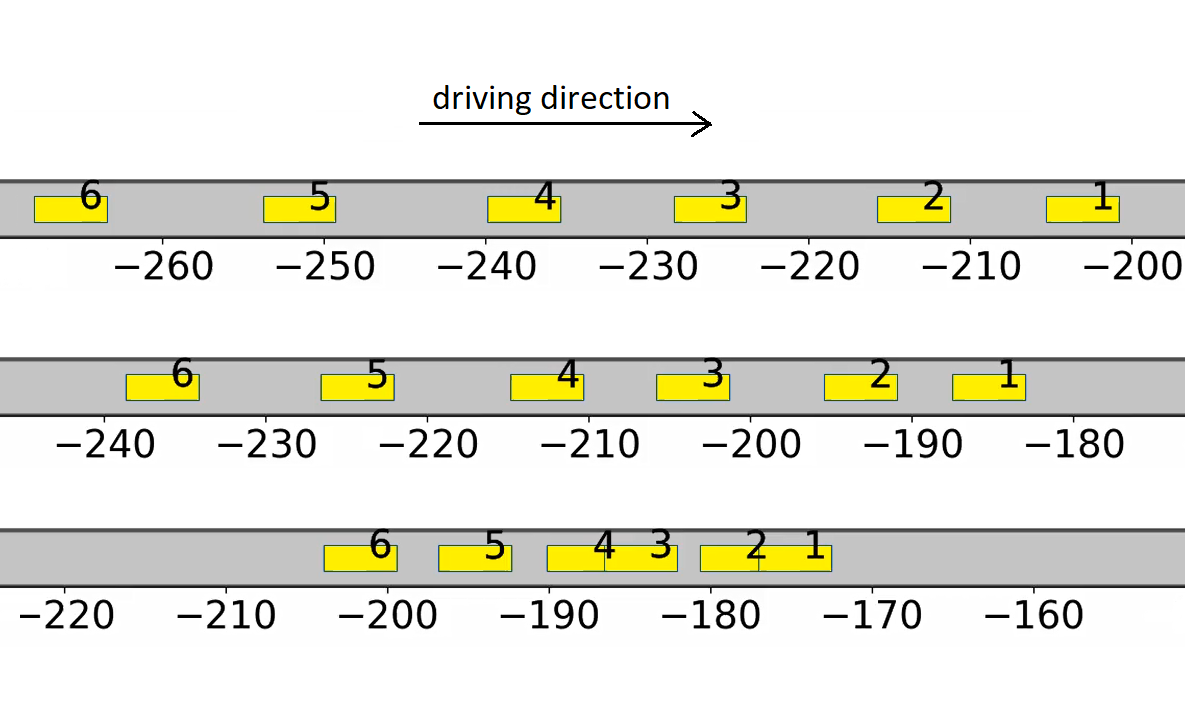}\vspace{-0.35cm}
  \caption{{\color{black}Snapshots from} CommonRoad without safety layer.}
  \label{CRASH_OFF_video}
	\end{subfigure}
\begin{subfigure}{0.48\textwidth}
	\centering
	\includegraphics[width=0.99\textwidth,trim=0cm 0cm 0cm 0.2cm,clip]{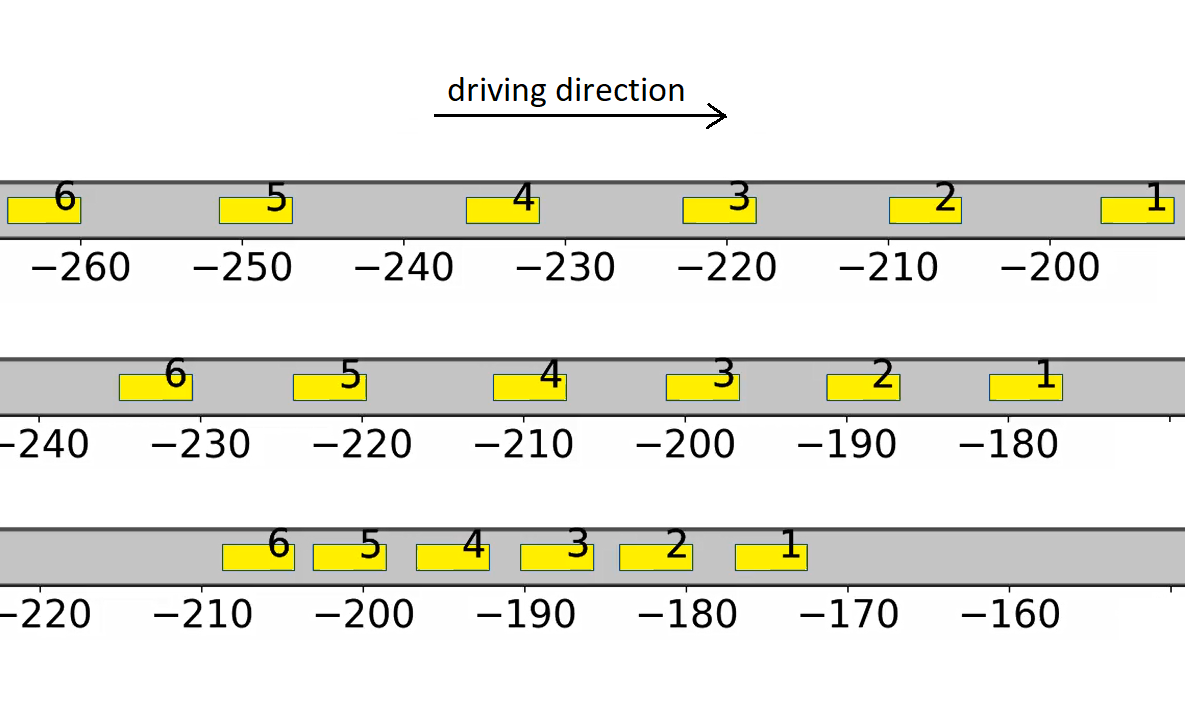}\vspace{-0.35cm}
	\caption{{\color{black}Snapshots from} CommonRoad with safety layer.}
	\label{CRASH_ON_video}
\end{subfigure}
~ %add desired spacing between images, e. g. ~, \quad, \qquad, \hfill etc. 
\caption{{\color{black}Spacing error, distance to leading vehicle, and snapshots from CommonRoad, with and without safety layer, for an emergency braking scenario.  When present, the provably-correct safety layer avoids the crash. The spacing policy is  small on purpose, to test safety in challenging scenarios.}}
\label{TOT_CRASH}
\end{figure*}

To demonstrate the benefits of the proposed resilient strategy \eqref{LTI_platoon_highgain}-\eqref{LTI_platoon_unkinput5}, we apply %consider the leader behavior as the previous simulation, augmented with 
the proposed observer. The observer design parameters are $S_a=S_j=1$, $\epsilon = 0.01$, $\eta = 1.5$, $Q_o = \diag\left(\left[0.1,0.2,0.01\right]\right)$, $\alpha_0 = 3$, $\alpha_1 = 0.2$, $\alpha_2 = 0.01$, % ({\color{black}resulting in} the Hurwitz polynomial $\lambda^3 + 3 \lambda^2+ 3 \lambda+ 1$ = $(\lambda+ 1)^3$), 
and $L_o$ selected to place the poles of $A_o-L_oC_o$ in $-5$ and $-1.5\pm j0.5$. For the leader behavior as before, Fig.~\ref{vel_comm_prop} shows that all vehicles manage to follow the velocity and acceleration of the leader (the red lines represent the observed variables from the observer). %Although it is well known that measurements noises degrade the observer performance \cite{KALSI2010347}, 
It can be noticed that the observer does not amplify the measurement noise. %s in the vehicle behavior. 
Remarkably, Fig.~\ref{error_comm_prop} shows that the spacing errors are kept close to zero, except for a small observer transient when the leader switches from acceleration to deceleration. Note that the transients in Fig.~\ref{error_comm_prop} are much smaller than those in Fig.~\ref{error_comm}, making the performance of the proposed resilient design in Fig.~\ref{TOT_comm_prop} clearly superior to that of Fig.~\ref{TOT_comm}.   %the  communication failure scenario. 

\subsection{Full-range safe-CACC evaluation with CommonRoad}

{\color{black}We evaluate the provably-safe control layer mentioned in Sec.~II-B via CommonRoad \cite{7995802}. CommonRoad is a collection of composable benchmarks for motion planning on roads, which provides a means of evaluating and comparing different control strategies for automated vehicles.}  
To evaluate the impact of self-organization, we perform tests with and without the self-organization mechanism: in these tests, the safety layer  is wrapped around the nominal platooning strategy, as discussed in Sec.~II-B. 
We test the presence of heterogeneous acceleration limits for a leader behavior resembling a stop-and-go behavior: the spacing errors  in Fig.~\ref{HET_OFF} show that the absence of self-organization causes spacing errors in the order of 10 meters, indicating poor platooning performance. Upon activating the proposed self-organization strategy, Fig.~\ref{HET_ON} shows that the spacing errors are small and converge close to zero. Snapshots from CommonRoad are shown in Fig.~\ref{HET_OFF_video} (without  self-organization) and Fig.~\ref{HET_ON_video} (with self-organization): note that the proposed self-organization strategy promotes a homogeneous behavior for all vehicles during the whole platooning operation, despite the heterogeneity in the vehicle dynamics. 
%It is worth remarking that no safety alarms are activated in the above tests during the operation of the proposed self-organization strategy: we have tested the proposed self-organization mechanism in CommonRoad for different initial conditions and the safety layer has never returned collision alarms, as long as the initial inter-vehicle distances are larger than the safety distance calculated by the safety layer.

To test the impact of the safety layer, we have created on purpose a challenging emergency-braking scenario: the scenario is challenging because the desired inter-vehicle spacing is set to be small ($h=0.3$). %, so that the vehicles are very close. 
By doing this, strong braking of the leading vehicle would activate the safety layer. Fig.~\ref{CRASH_OFF} shows that the absence of a safety layer can indeed lead to a crash, cf. the snapshots from CommonRoad in Fig.~\ref{CRASH_OFF_video}. However, when the safety layer is active, it promptly activates emergency braking to avoid a collision, as shown % profile, leading to collision avoidance. Refer to the snapshots from CommonRoad  
in Fig.~\ref{CRASH_ON_video}. %Overall, the crash tests in Fig.~\ref{TOT_CRASH} validate the effectiveness of the safety layer in handling non-nominal platooning operation.

\subsection{Numerical experiments with CARLA}

{\color{black}To further evaluate our work, we make use of the toolbox OpenCDA \cite{xu2021opencda} in the CARLA simulator (cf. Fig.~\ref{11b}). 
We consider a platoon of one leading vehicle and four following vehicles, where the type of vehicles are: Lincoln MKZ2017, Dodge Charger2020, BMW Grand Tourer, Ford Mustang, Nissan Patrol.} Upon estimating approximately the engine time constant of each vehicle, we implemented the homogenizing strategy in Sec.~III, so as to compare the difference between the absence and the presence of a common group model. In order to test the strategy in different conditions, we considered three different behaviors for the leading vehicle:

\begin{figure*}[t]
	\centering
	\begin{subfigure}{0.49\textwidth}
		\centering
		\includegraphics[width=0.99\textwidth]{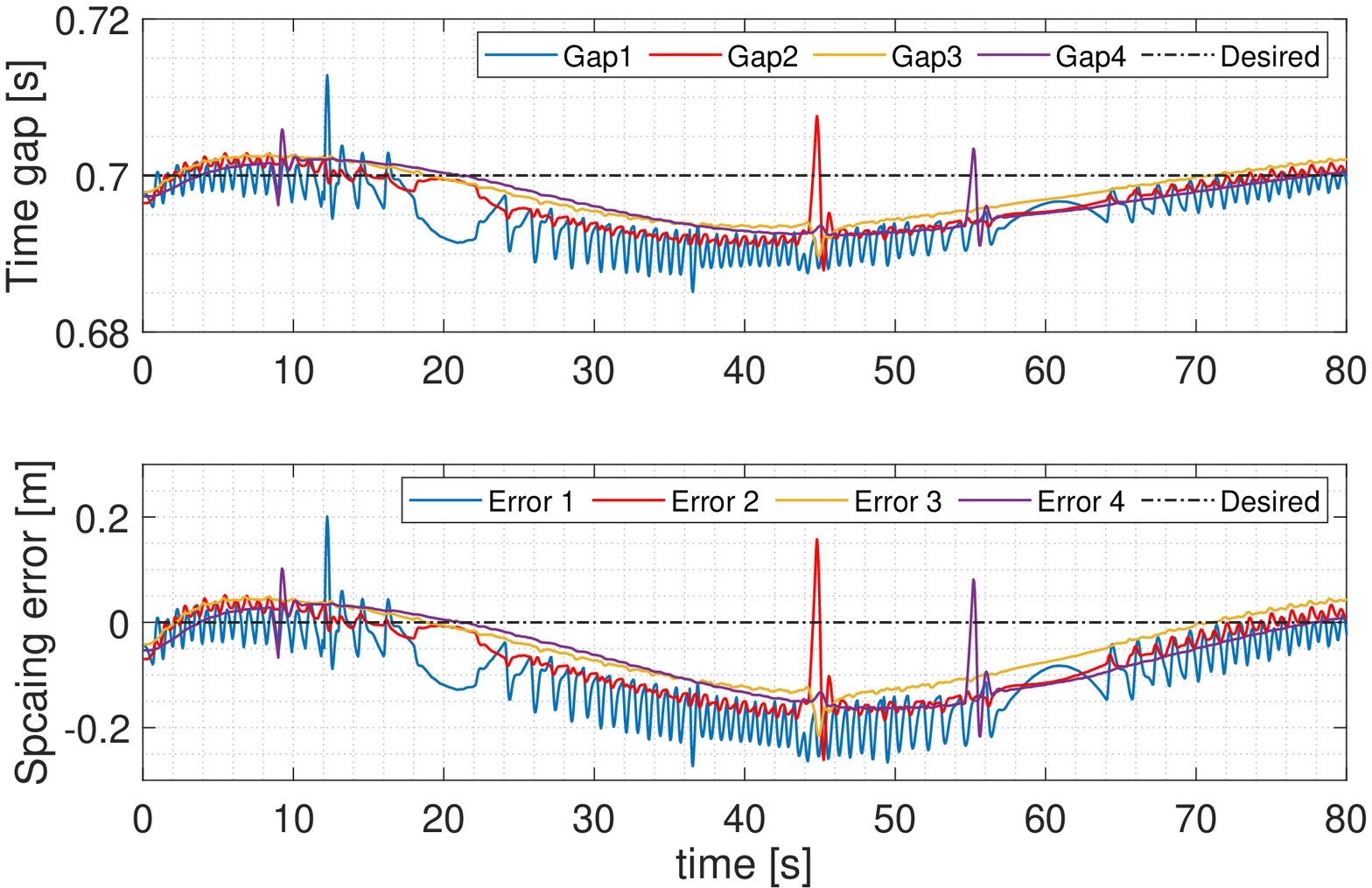}
		\caption{Time gap (upper plot) and spacing error (lower plot) \newline for all vehicles,  without group model.}
		\label{CARLAvel1}
	\end{subfigure}
	~ %add desired spacing between images, e. g. ~, \quad, \qquad, \hfill etc. 
	\begin{subfigure}{0.49\textwidth}
		\centering
		% include second image
		\includegraphics[width=0.99\textwidth]{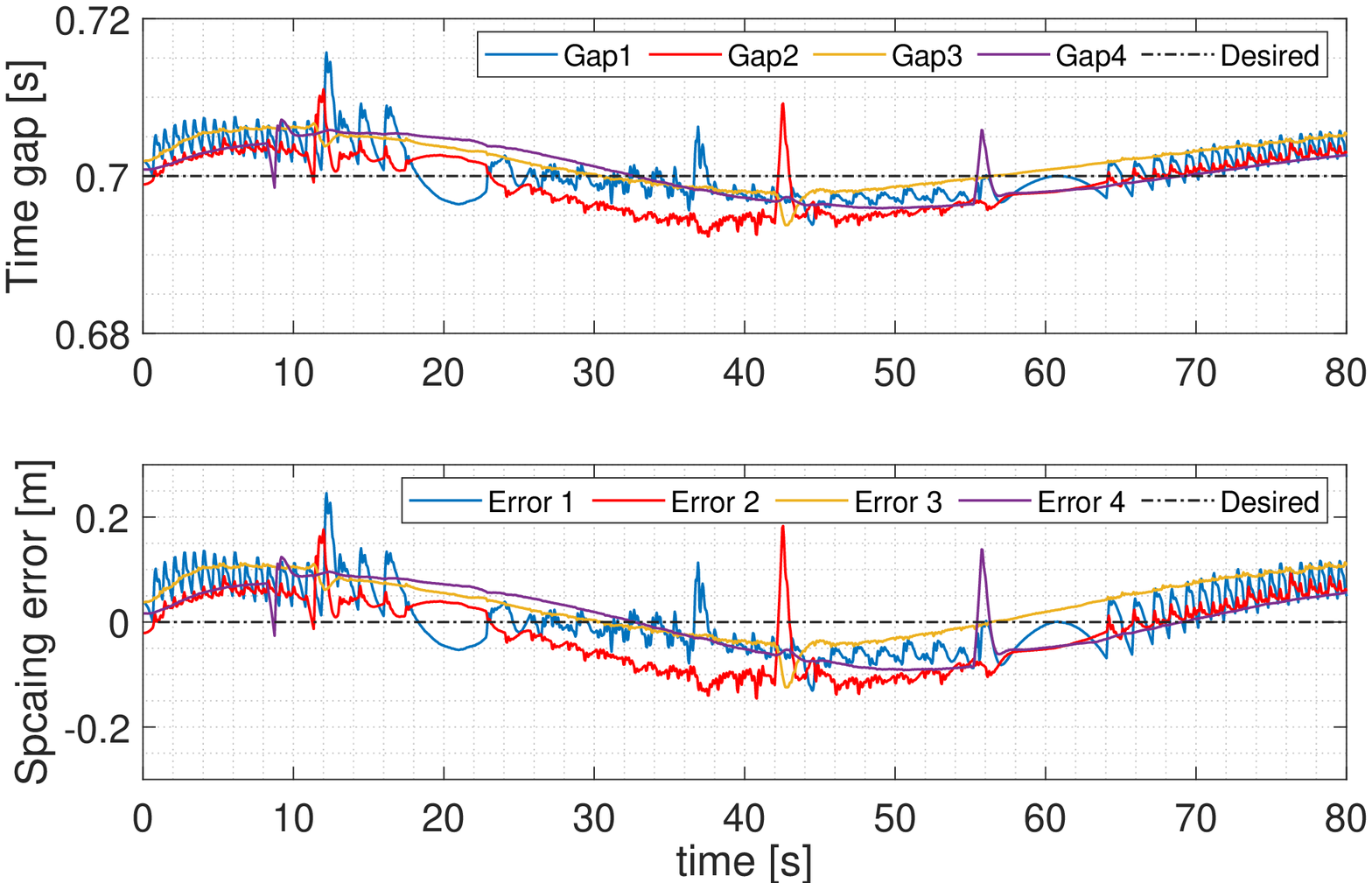}
		\caption{Time gap (upper plot) and spacing error (lower plot) \newline for all vehicles, with group model.}
		\label{CARLAvel2}
	\end{subfigure}
	\caption{Results with CARLA simulator for Scenario 1.}
	\label{TOT_CARLAvel}
\end{figure*}

\begin{figure}[t]
	\centering
	%\begin{subfigure}{0.5\textwidth}
	%\centering
	\vspace{-0.5cm}\includegraphics[width=0.48\textwidth]{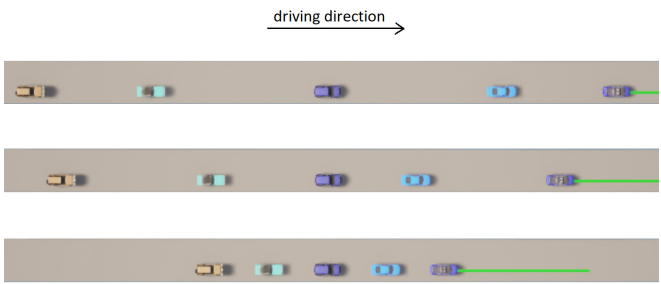}\vspace{-0.2cm}
	\caption{Snapshots from OpenCDA in CARLA.}
	\label{11b}
\end{figure}

\begin{table*}[b]
	%\captionsetup{width=.75\textwidth}
	\captionsetup{width=.99\textwidth}
	\caption{\text{Results of numerical experiments in CARLA. % in different scenarios. 
	The time gap error is evaluated between adjacent vehicles.}}
	\label{CARLAtab}
	\centering
	%\begin{center}
	{\footnotesize
		\begin{tabular}{llllllllll}
			\toprule
			%&$\quad$& Scenario 1 & Scenario 1 &$\quad$& Scenario 2 & Scenario 2 &$\quad$& Scenario 3 & Scenario 3\\
			&$\quad$& \multicolumn{2}{c}{Scenario 1} &$\quad$& \multicolumn{2}{c}{Scenario 2} &$\quad$& \multicolumn{2}{c}{Scenario 3}\\
			&$\quad$& no group model & with group model &$\quad$& no group model & with group model &$\quad$& no group model & with group model \\
			\midrule
			 time gap error 1 &$\quad$& 0.0946 & \textbf{0.0578} $\ $(-63.7\%) &$\quad$& 0.1060 & \textbf{0.0726} $\ $(-46.0\%) &$\quad$& 0.1365  & \textbf{0.1076} $\ $(-26.9\%)\\
			 time gap error 2 &$\quad$& 0.0723 & \textbf{0.0597} $\ $(-21.1\%) &$\quad$& 0.0851 & \textbf{0.0739} $\ $(-15.2\%) &$\quad$& 0.1155 & \textbf{0.1097} $\ $$\ $(-5.3\%)\\
			 time gap error 3 &$\quad$& 0.0613 & \textbf{0.0587} $\ $$\ $(-4.4\%) &$\quad$& 0.0753 & \textbf{0.0672} $\ $(-12.1\%) &$\quad$& 0.1136 & \textbf{0.1005} $\ $(-13.0\%) \\
			 time gap error 4 &$\quad$& 0.0664 & \textbf{0.0539} $\ $(-23.2\%) &$\quad$& 0.0734 & \textbf{0.0650} $\ $(-12.9\%) &$\quad$& 0.1032  & \textbf{0.0900} $\ $(-14.7\%)\\
			%\hline
			\bottomrule
		\end{tabular}
	}
	%\end{center}
\end{table*}
%\clearpage

\begin{itemize}
	\item Scenario 1: the velocity of the leading vehicle oscillates around 15m/s with amplitude 6m/s and period 70s;
	\item Scenario 2: the velocity of the leading vehicle oscillates around 15m/s with amplitude 6m/s and period 50s;
	\item Scenario 3: the velocity of the leading vehicle oscillates around 15m/s with amplitude 6m/s and period 30s.
\end{itemize}
%Considering different periods allow to test the response of the strategy. 
%For all scenarios, the amplitude of the oscillation was 6m/s. 
The results are reported in Tab.~\ref{CARLAtab} in terms of the time gap error between adjacent vehicles (the difference between the desired time gap and the actual time gap). The results are also reported in Fig.~\ref{TOT_CARLAvel} in terms of time gap and spacing error (for compactness, only the results for Scenario 1 are reported). From the plots it is possible to see that the main difference between Fig.~\ref{CARLAvel1} (without common group model) and Fig.~\ref{CARLAvel2} (with common group model) is that the presence of a common group model attenuates the oscillations around the desired spacing. %Such oscillations are induced by the oscillating behavior of the leading vehicle: in the absence of a group model, the oscillations are more pronounced and different for different vehicles due to the different vehicle types; when reaching a common group model thanks to the homogenizing strategy, the oscillations are smaller and close to each other. 
This shows the effectiveness of enforcing a common behavior to improve platooning performance.

\subsection{Numerical experiments in complex traffic scenarios}

{\color{black}We finally use CommonRoad to test complex traffic scenarios beyond the standard platooning operation: %. In particular, we test:
\begin{itemize}
\item a merging scenario with three different platoons merging into one;
\item a cut-in scenario with one vehicle suddenly cutting in the middle of the platoon.
\end{itemize}}
These scenarios are tested for the proposed self-organization mechanism and the provably-correct safety layer wrapped around it. The merging scenario has two phases, as shown in Fig.~\ref{MERGE} and Fig.~\ref{MERGE_video}: first the two platoons on the left will merge, and then the vehicles will all accelerate to merge with the platoon on the right. The maneuver is accomplished successfully. In the cut-in scenario, the velocity and spacing errors are in Fig.~\ref{CUT}: the red vehicle shown in Fig.~\ref{CUT_video} suddenly cuts in the middle of the platoon violating the safety distance. As a result, the safety layer is activated, in particular the safety recapturing controller that aims to recover the safety distance in finite time: this explains the large spacing error in Fig.~\ref{CUT}, which is needed to create the gap for the cut-in vehicle to safely merge.  
Summarizing, the tests in Figs.~\ref{TOT_COMPLEX} and \ref{TOT_COMPLEX2} suggest that the proposed self-organization mechanism can be embedded into a full-range safe CACC protocol. % with safety characteristics.

\begin{figure}[t]
	\centering
	\begin{subfigure}{0.5\textwidth}
		\centering
		\includegraphics[width=0.99\textwidth]{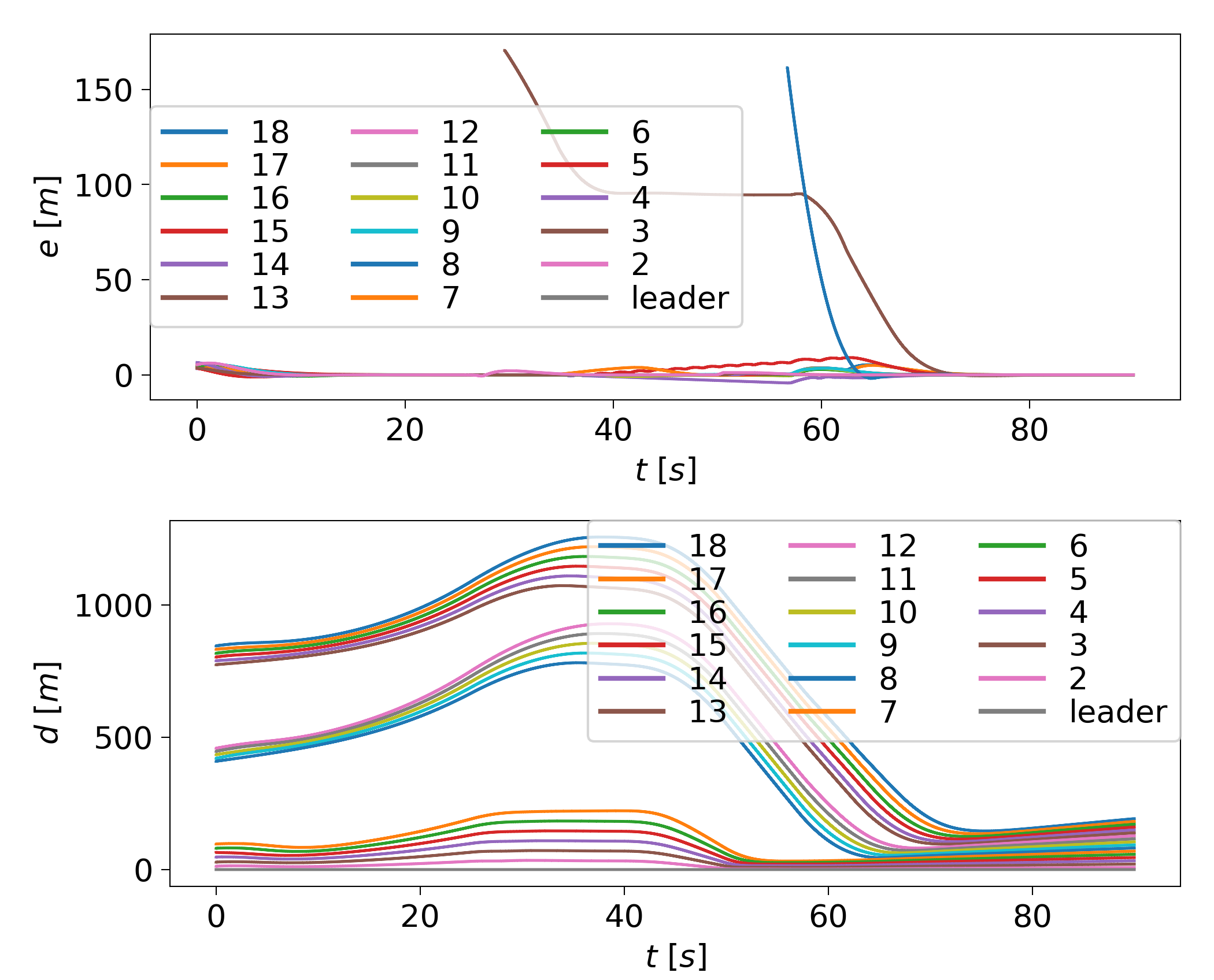}
		\caption{Spacing error and distance to leading vehicle.}
		\label{MERGE}
	\end{subfigure}
	~ %add desired spacing between images, e. g. ~, \quad, \qquad, \hfill etc. 
	\begin{subfigure}{0.5\textwidth}
		\centering
		\vspace{0.3cm}
		\includegraphics[width=0.99\textwidth]{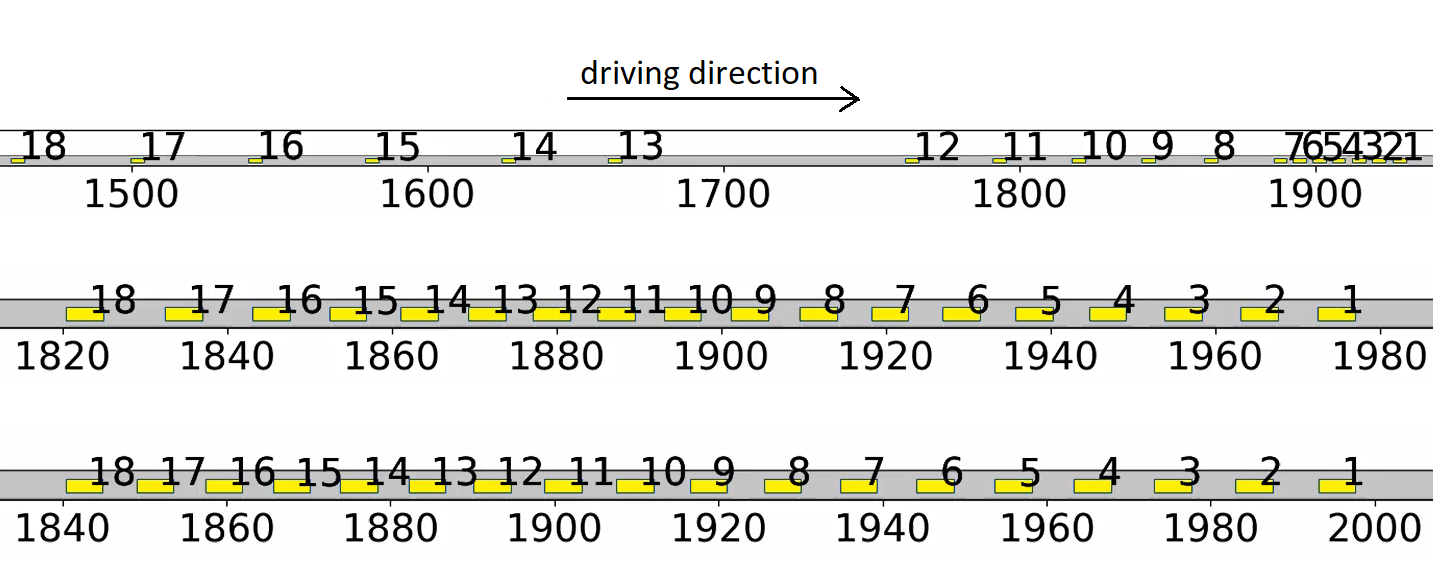}\vspace{0.3cm}		
		\caption{Snapshots from CommonRoad {\color{black}(vehicles appear with different size due to the different length of the road in each snapshot)}.}
		\label{MERGE_video}
	\end{subfigure}
	\caption{{\color{black}Spacing error, distance to leading vehicle, and snapshots from CommonRoad, for a merging scenario. The proposed self-organizing strategy and the safety layer are both included. The scenario is accomplished successfully.}}
	\label{TOT_COMPLEX}
\end{figure}

\begin{figure}[t]
	\centering
	\begin{subfigure}{0.5\textwidth}
		\centering
		\includegraphics[width=0.99\textwidth]{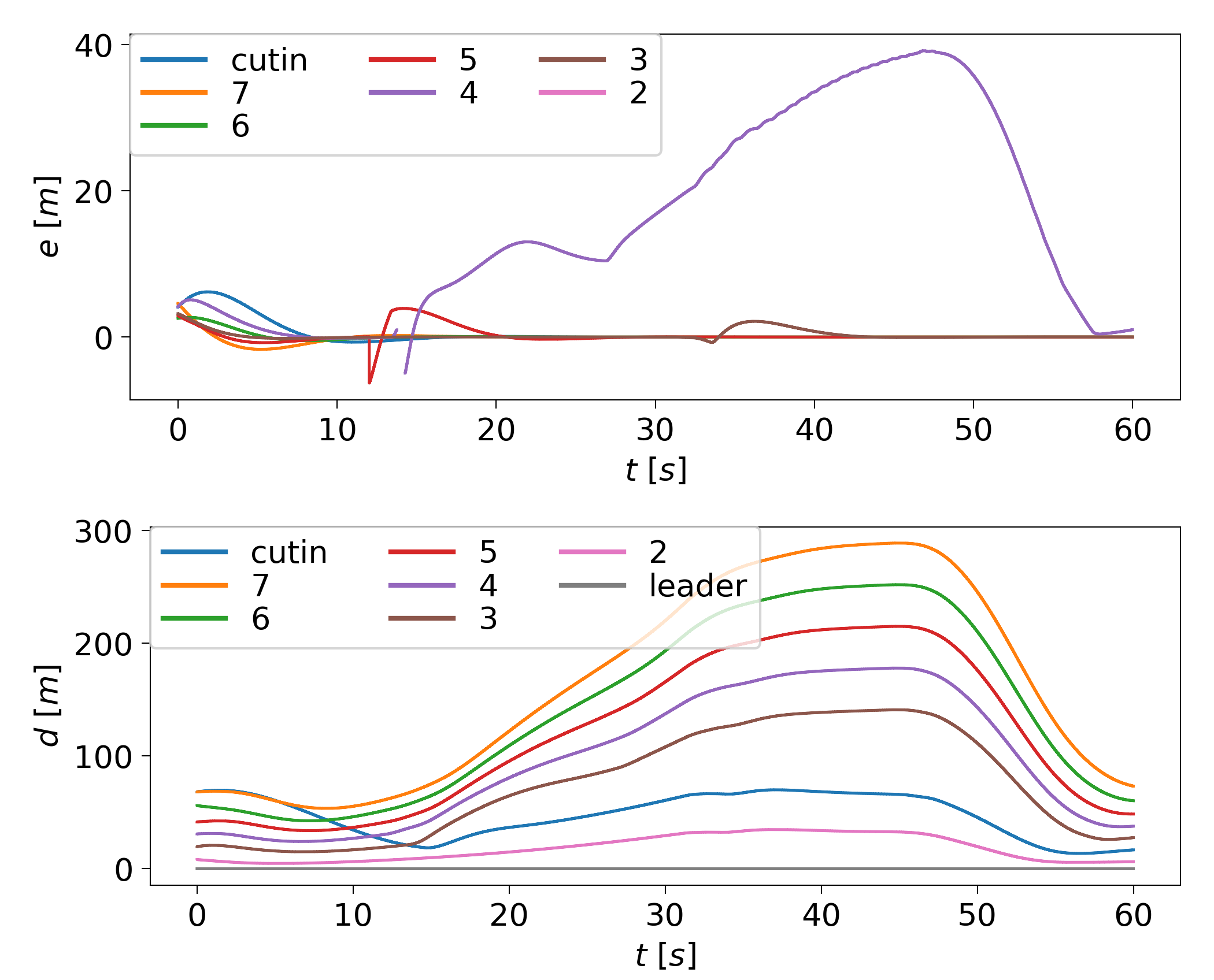}
		\caption{Spacing error and distance to leading vehicle.}
		\label{CUT}
	\end{subfigure}
	~ %add desired spacing between images, e. g. ~, \quad, \qquad, \hfill etc. 
	\begin{subfigure}{0.5\textwidth}
		\centering
		\includegraphics[width=0.99\textwidth]{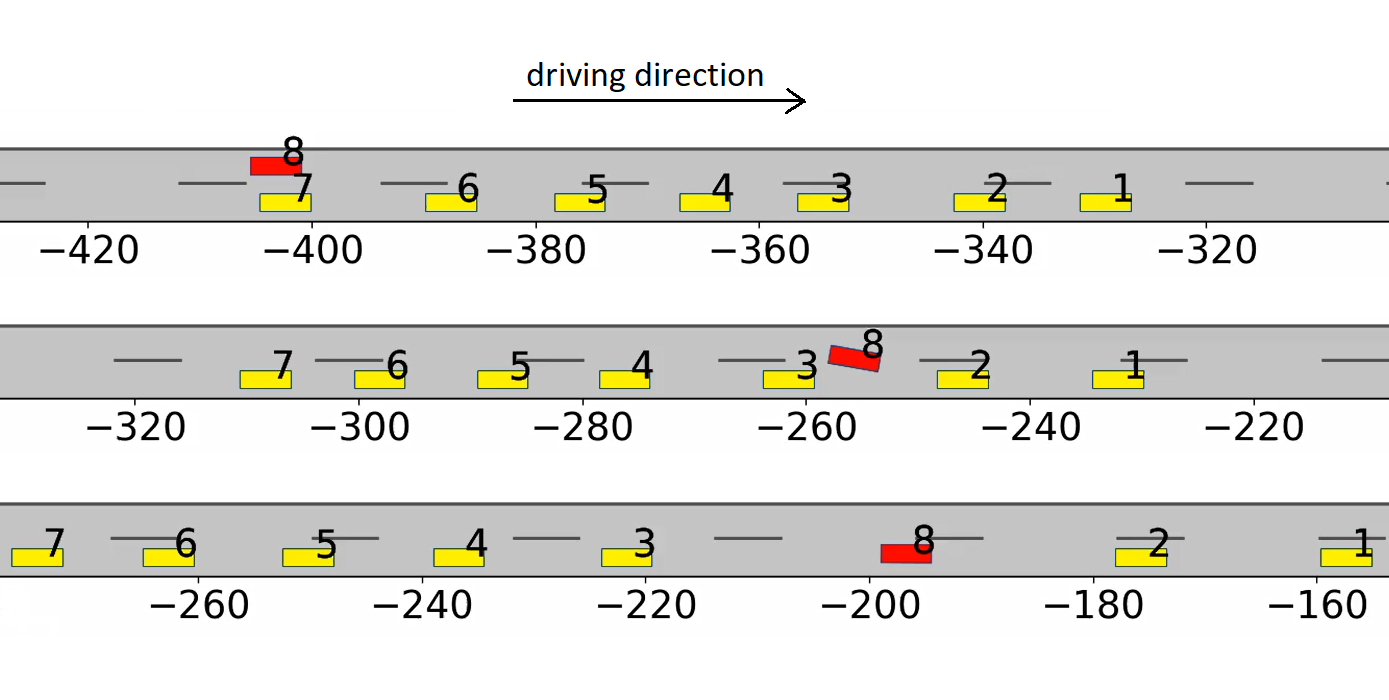}
		\caption{Snapshots from CommonRoad.}
		\label{CUT_video}
	\end{subfigure}
	\caption{{\color{black}Spacing error, distance to leading vehicle, and snapshots from CommonRoad, for a cut-in scenario. The proposed self-organizing strategy and the safety layer are both included. The scenario is accomplished successfully.}}
	\label{TOT_COMPLEX2}
\end{figure}

\section{Conclusions and future work}\label{Chap_conclusion}

This work studied self-organization strategies in formations of cooperative vehicles and their provably-correct safety implementation. Self-organization allows all vehicles to homogenize to a common (group) behavior, which has been shown to promote resilience against {\color{black}acceleration} limits and wireless communication failures. 
Stability of the self-organization mechanism has been studied analytically in terms of convergence of the spacing errors to zero and convergence to a common homogeneous behavior, despite the vehicles comprising the platoon being heterogeneous. The theoretical results have been supported by numerical implementations in the platooning toolbox OpenCDA in CARLA and in the CommonRoad platform. The proposed self-organization approach was integrated in a provably-correct layer to ensure safety. %wrapped around the nominal platooning protocol and active during non-nominal operation like emergency braking and cut-in vehicles.

{\color{black}Interesting topics for future work are to handle nonlinear vehicle dynamics by possibly combining self-organization with nonlinear control (e.g., backstepping), and handle uncertain vehicle dynamics by possibly combining self-organization with adaptive control \cite{10130071}.} Also, stemming from the observation that platoons require some homogeneity to operate effectively \cite{7466806}, self-organization can be useful to determine metrics for vehicles leaving or joining a platoon. 

%In future work, self-organization can be used to determine metrics for vehicles leaving or joining a platoon. These metrics stem from the observation that platoons require some homogeneity to operate effectively \cite{7466806} (e.g. it is not effective for a sport car to join a platoon of trucks). Motivated by the proposed validations with a provably-correct safety layer, another interesting work is to study the benefits of CACC as compared to ACC in terms of safety.

%%%%%%%%%%%%%%%%%%%%%%%%%%%%%%%%%%%%%%%%%%%%%%%%%%%%%%%%%%%%%%%%%
%%%%%%%%%%%%%%%%%%%%%%%%%%%%%%%%%%%%%%%%%%%%%%%%%%%%%%%%%%%%%%%%%

\ifCLASSOPTIONcaptionsoff
  \newpage
\fi

\bibliographystyle{IEEEtran}
\bibliography{SelforgCACC_final}

\end{document}